\documentclass[preprint,12pt]{elsarticle}
\usepackage[T1]{fontenc}
\journal{Theoretical Computer Science}

\usepackage{comment}
\newcommand\thankssymb[1]{\textsuperscript{\@fnsymbol{#1}}}

\usepackage{hyperref}

\usepackage{mathrsfs}
\usepackage{xcolor}
\usepackage{algorithm,needspace}
\usepackage{algcompatible}
\usepackage[noend]{algpseudocode}
\MakeRobust{\Call} %algoritmo
\newcommand*\Let[2]{\State #1 $\gets$ #2}
\algrenewcommand\algorithmicrequire{\textbf{Requisito:}}
\algrenewcommand\algorithmicensure{\textbf{Observação:}}
\algrenewcommand\algorithmicfunction{\textbf{function}}
\algrenewcommand\algorithmicfor{\textbf{for}}
\algrenewcommand\algorithmicdo{\textbf{do}}
\algrenewcommand\algorithmicend{\textbf{end}}
\algrenewcommand\algorithmicreturn{\textbf{return}}
\algrenewcommand\algorithmicprocedure{\textbf{procedure}}
\algrenewcommand\algorithmicif{\textbf{if}}
\algrenewcommand\algorithmicelse{\textbf{else}}
\algrenewcommand\algorithmicthen{\textbf{then}}
\algrenewcommand\algorithmicwhile{\textbf{while}}
\algdef{SE}[DOWHILE]{Do}{doWhile}{\algorithmicdo}[1]{\algorithmicwhile\ #1}%

\makeatletter
\renewcommand*{\ALG@name}{Algorithm}

\newenvironment{nofloatalgorithmic}[3][0]
  {% \begin{nofloatalgorithmic}
  \par
  \def\algoLabel{#3}
  \addvspace{\intextsep}% Vertical gap above in-text float
  \needspace{\dimexpr\baselineskip+6.8pt}%
  \noindent%
  \hrule height.8pt depth0pt \kern2pt% Top horizontal rule
  \refstepcounter{algorithm}% Step algorithm counter
  \addcontentsline{loa}{algorithm}{\numberline{\thealgorithm}#2}% Add entry to LoA
  \noindent\textbf{\fname@algorithm~\thealgorithm} #2\par% Write caption
  \kern2pt\hrule\kern2pt% Middle horizontal rule
  \begin{algorithmic}[#1]
  }
  {% \end{nofloatalgorithmic}
  \end{algorithmic}\label{\algoLabel}
  \nobreak\kern2pt\hrule\relax% Bottom horizontal rule
  \addvspace{\intextsep}% Vertical gap below in-text float
  }
\makeatother

\newcommand{\LetVertical}[2]{#1 $\gets$ #2}

\usepackage{support-caption}
\usepackage{amsthm}

\usepackage{enumerate}
\usepackage{subfig}
\usepackage{lipsum}
\usepackage{svg}

\usepackage{amsmath}

\usepackage{amsthm}
\newtheorem{theorem}{Theorem}
\newtheorem{lemma}{Lemma}
\newtheorem{corollary}{Corollary}
\newtheorem{definition}{Definition}

\usepackage{amssymb}
\usepackage{makecell}
\usepackage[title]{appendix}
\usepackage{tikz}
\usepackage{tkz-berge}

\usepackage{pbox}

\usepackage{complexity}
\newclass{\Hard}{hard}
\newclass{\pNP}{paraNP}
\newclass{\Hness}{hardness}
\newcommand{\NPH}{\NP\text{-}\Hard}
\newcommand{\pNPH}{\pNP\text{-}\Hard}

\newclass{\Complete}{complete}
\newclass{\Cness}{completeness}
\newcommand{\NPc}{\NP-\Complete}

\newcommand{\NPcness}{\NP-\Cness}
\newfunc{\YES}{YES}
\newfunc{\NOi}{NO}
\newfunc{\tw}{tw}
\newfunc{\sift}{ref}
\newcommand{\pname}[1]{\textsc{#1}}
\newcommand{\WH}[1]{\W[#1]\text{-}\Hard}

\newcommand{\bigO}[1]{\mathcal{O}\!\left(#1\right)}
\newcommand{\nproblem}[3]{{\centering\fbox{\pbox{\textwidth}{\pname{#1}\\\textit{Instance}: #2\\\textit{Question}: #3}}}}

\begin{document}

\begin{frontmatter}

%% Title, authors and addresses

%% use the tnoteref command within \title for footnotes;
%% use the tnotetext command for theassociated footnote;
%% use the fnref command within \author or \address for footnotes;
%% use the fntext command for theassociated footnote;
%% use the corref command within \author for corresponding author footnotes;
%% use the cortext command for theassociated footnote;
%% use the ead command for the email address,
%% and the form \ead[url] for the home page:
%% \title{Title\tnoteref{label1}}
%% \tnotetext[label1]{}
%% \author{Name\corref{cor1}\fnref{label2}}
%% \ead{email address}
%% \ead[url]{home page}
%% \fntext[label2]{}
%% \cortext[cor1]{}
%% \affiliation{organization={},
%%             addressline={},
%%             city={},
%%             postcode={},
%%             state={},
%%             country={}}
%% \fntext[label3]{}

\title{Disconnected Matchings}

%% use optional labels to link authors explicitly to addresses:
%% \author[label1,label2]{}
%% \affiliation[label1]{organization={},
%%             addressline={},
%%             city={},
%%             postcode={},
%%             state={},
%%             country={}}
%%
%% \affiliation[label2]{organization={},
%%             addressline={},
%%             city={},
%%             postcode={},
%%             state={},
%%             country={}}

\author[inst1]{Guilherme C. M. Gomes}
\ead{gcm.gomes@dcc.ufmg.br}
\author[inst2]{Bruno P. Masquio\corref{cor1}}
\ead{brunomasquio@ime.uerj.br}
\author[inst2]{Paulo E. D. Pinto}
\ead{pauloedp@ime.uerj.br}
\author[inst1]{Vinicius F. dos Santos}
\ead{viniciussantos@dcc.ufmg.br}
\author[inst2,inst3]{Jayme L. Szwarcfiter}
\ead{jayme@nce.ufrj.br}
\cortext[cor1]{Corresponding author}

%\thanks{\thankssymb{1} for all the fish}
%\tnotetext[t2]{The second title footnote which is a longer longer than the first one and with an intention to fill in up more than one line while formatting.}

\affiliation[inst1]{organization={Departamento de Ciência da Computação - Universidade Federal de Minas Gerais (UFMG)},%Department and Organization
            %addressline={Address One}, 
            city={Belo Horizonte},
            %postcode={00000},
            %state={Minas Gerais},
            country={Brazil},
            }

\affiliation[inst2]{organization={Instituto de Matemática e Estatística - Universidade do Estado do Rio de Janeiro (UERJ)},%Department and Organization
            %addressline={Address Two}, 
            city={Rio de Janeiro},
            %postcode={22222}, 
            %state={Rio de Janeiro},
            country={Brazil}}
            
\affiliation[inst3]{organization={Instituto de Matemática e PESC/COPPE - Universidade Federal do Rio de Janeiro (UFRJ)},%Department and Organization
%addressline={Address Two}, 
city={Rio de Janeiro},
%postcode={22222}, 
%state={Rio de Janeiro},
country={Brazil}}

\begin{abstract}
In 2005, Goddard, Hedetniemi, Hedetniemi and Laskar [Generalized subgraph-restricted matchings in graphs, Discrete Mathematics, 293 (2005) 129 – 138] asked the computational complexity of determining the maximum cardinality of a matching whose vertex set induces a disconnected graph.
In this paper we answer this question. In fact, we consider the generalized problem of finding \emph{$c$-disconnected matchings}; such matchings are ones whose vertex sets induce subgraphs with at least $c$ connected components.
We show that, for every fixed $c \geq 2$, this problem is {\NPc} even if we restrict the input to bounded diameter bipartite graphs, while can be solved in polynomial time if $c = 1$. For the case when $c$ is part of the input, we show that the problem is {\NPc} for chordal graphs, while being solvable in polynomial time for interval graphs.
Finally, we explore the parameterized complexity of the problem.
We present an {\FPT} algorithm under the treewidth parameterization, and an {\XP} algorithm for graphs with a polynomial number of minimal separators when parameterized by $c$.
We complement these results by showing that, unless $\NP \subseteq \coNP/\poly$, the related \textsc{Induced Matching} problem does not admit a polynomial kernel when parameterized by vertex cover and size of the matching nor when parameterized by vertex deletion distance to clique and size of the matching. As for \textsc{Connected Matching}, we show how to obtain a maximum connected matching in linear time given an arbitrary maximum matching in the input.
\end{abstract}

%%Graphical abstract
%\begin{graphicalabstract}
%\includegraphics{grabs}
%\end{graphicalabstract}

%%Research highlights
%\begin{highlights}
%\item Research highlight 1
%\item Research highlight 2
%\end{highlights}

\begin{keyword}

%% keywords here, in the form: keyword \sep keyword
Algorithms \sep Complexity \sep Induced Subgraphs \sep Matchings
%% PACS codes here, in the form: \PACS code \sep code
%\PACS 0000 \sep 1111
%% MSC codes here, in the form: \MSC code \sep code
%% or \MSC[2008] code \sep code (2000 is the default)
%\MSC 0000 \sep 1111
\end{keyword}

\end{frontmatter}

%% \linenumbers

\section{Introduction}

Matchings are a widely studied subject both in structural and algorithmic graph theory~\cite{edmonds_matching,10.1007/s00453-003-1035-4,LOZIN20027,BrunoMasquio:2019:TeseMestrado,vazirani,MOSER2009715,10.1007/978-3-030-48966-3_31}.
A matching is a subset $M \subseteq E$ of edges of a graph $G = (V,E)$ that do not share any endpoint.
A $\mathscr{P}$-matching is a matching such that $G[M]$, the subgraph of $G$ induced by the endpoints of edges of $M$, satisfies property $\mathscr{P}$.
The complexity of deciding whether or not a graph admits a $\mathscr{P}$-matching has been investigated for many different properties $\mathscr{P}$ over the years. One of the most well known examples is the {\NPcness} of the \pname{Induced Matching} problem~\cite{CAMERON198997}, where $\mathscr{P}$ is the property of being a 1-regular graph.
Little is known about structural parameterizations for \pname{Induced Matching}, and even less about kernelization.
In~\cite{MOSER2009715}, Moser and Sikdar present a series of \FPT\ algorithms parameterized by the size of the matching for various graph classes, including planar graphs, bounded degree graphs, and line graphs; they also present a linear kernel under this parameterization for planar graphs.
Another commonly studied parameter is the vertex deletion distance to a matching, i.e. the minimum number of vertices that must be removed from the graph to obtain a 1-regular graph.
A corollary of the work of Moser and Thilikos~\cite{cubic_kernel_induced_matching} on regular graphs yields a cubic kernel for \pname{Induced Matching} under this parameterization; this was later improved to a quadratic kernel by Mathieson and Szeider~\cite{quadratic_kernel_induced_matching} and, more recently, to a linear kernel by Xiao and Kou~\cite{linear_kernel_induced_matching}.

Other {\NPH} problems include \pname{Acyclic Matching}~\cite{GODDARD2005129}, \pname{$k$-Degenerate Matching}~\cite{BASTE201838}, deciding if the subgraph induced by a matching contains a unique maximum matching~\cite{Golumbic2001}, and \pname{Line-Complete Matching}\footnote{A line-complete matching $M$ is a matching such that every pair of edges of $M$ has a common adjacent edge.}~\cite{cameron_connected_matching}; the latter was originally named \pname{Connected Matching}, but we adopt the more recent meaning of \pname{Connected Matching} given by Goddard et. al~\cite{GODDARD2005129}, where we want the subgraph induced by the matching to be connected.
We summarize the above results in Table~\ref{tab:p-matchings}.

%A graph $G$ is $k$-degenerated if the minimum degree of $H \subseteq G$, $|V(H)| > 0$, is at most $k$.

\begin{table}[!htb]
\begin{center}
\begin{tabular}{c|c|c}
\hline
$\mathscr{P}$-matching & \makecell{Property $\mathscr{P}$} & Complexity \\ \hline
{ \sc Induced Matching }                                                & $1$-regular               & \makecell{{\NPc} \cite{CAMERON198997}} \\ \hline
{ \sc Acyclic Matching }                                                & acyclic                  & \makecell{{\NPc} \cite{GODDARD2005129}} \\ \hline
{ \sc $k$-Degenerate Matching }                                          & $k$-degenerate            & \makecell{{\NPc} \cite{BASTE201838}} \\ \hline
\makecell{{ \sc Uniquely Restricted} \\ { \sc  Matching }}                                    & \makecell{has a unique \\ maximum \\ matching}       & \makecell{{\NPc} \cite{Golumbic2001}} \\ \hline
{ \sc Connected Matching }                                                  & connected                    & \makecell{Polynomial~\cite{GODDARD2005129} \\ Same as \\ {\sc Maximum Matching}$^\dagger$  } \\ \hline
\makecell{{\sc $c$-Disconnected Matching},  \\ for each $c \geq 2$}                                               & \makecell{has $c$ \\ connected \\ components}                 & \makecell{\makecell{{\NPc} \\ for bipartite graphs$^\dagger$}} \\ \hline
\makecell{{\sc Disconnected Matching},  \\ with $c$ as part of the input}                                    & \makecell{has $c$ \\ connected \\ components}                 & \makecell{\makecell{{\NPc} \\ for chordal graphs$^\dagger$}} \\ \hline
\end{tabular}
\end{center}
\caption{$\mathscr{P}$-matchings and some of its complexity results. Entries marked with a $\dagger$ are presented in this paper.} \label{tab:p-matchings}
\end{table}

Motivated by a question posed by Goddard et al.~\cite{GODDARD2005129} about the complexity of finding a matching that induces a disconnected graph, in this paper we study the \pname{Disconnected Matching} problem, which we define as follows:

\nproblem{Disconnected Matching}{A graph $G$ and two integers $k$ and $c$.}{Is there a matching $M$ with at least $k$ edges such that $G[M]$ has at least $c$ connected components?}

Our first result is an alternative proof for the polynomial time solvability of \pname{Connected Matching}.
We then answer Goddard et al.'s question by showing that \pname{Disconnected Matching} is {\NPc} for $c=2$. Indeed, we show that the problem remains {\NPc} even on bipartite graphs of diameter three, for every fixed $c \geq 2$; we denote this version of the problem by \pname{$c$-Disconnected Matching}.
Note that, while \pname{Induced Matching} is the particular case of \pname{Disconnected Matching} when $c = k$, our result is much more general since we decouple these two parameters.
Then, we turn our attention to the complexity of this problem on graph classes and parameterized complexity.
We begin by showing that, unlike \pname{Induced Matching}, \pname{Disconnected Matching} remains {\NPc} even when restricted to chordal graphs of diameter $3$; in this case, however, $c$ is part of the input, and we also prove that, for every fixed $c$, we can solve the problem in {\XP} time.
Afterwards, we present a polynomial time dynamic programming algorithm for interval graphs.
We then focus on the parameterized complexity of \pname{Disconnected Matching}.
In this context, we first show an {\FPT} algorithm parameterized by treewidth, then proceed to explore kernelization aspects of the problem.
Using the cross-composition framework~\cite{cross_composition}, we show that, unless $\NP \subseteq \coNP/\poly$, \pname{Induced Matching} and, consequently, \pname{Disconnected Matching}, do not admit polynomial kernels when parameterized by vertex cover and size of the matching nor when parameterized by vertex deletion distance to clique and size of the matching.
%Little is known about structural parameterizations for \pname{Induced Matching}, and even less about kernelization.
%In~\cite{MOSER2009715}, Moser and Sikdar present a series of \FPT\ algorithms parameterized by the size of the matching for various graph classes, including planar graphs, bounded degree graphs, and line graphs.
%Another commonly studied parameter is the vertex deletion distance to a matching, i.e. the minimum number of vertices that must be removed from the graph to obtain a 1-regular graph.
%A corollary of the work of Moser and Thilikos~\cite{cubic_kernel_induced_matching} on yields a cubic kernel for \pname{Induced Matching} under this parameterization; this was later improved to a quadratic kernel by Mathieson and Szeider~\cite{quadratic_kernel_induced_matching} and, more recently, to a linear kernel by Xiao and Kou~\cite{linear_kernel_induced_matching}.
We summarize our complexity results in Table~\ref{tab:disc-p-matching}.

\begin{table}[!htb]
\begin{center}
\begin{tabular}{c|c|c|c}
\hline
Graph class & $c$ & Complexity & Proof \\ \hline
General                                                & $c = 1$               & \makecell{Same as \pname{Maximum Matching}} & Theorem \ref{teo:emp-1-desc} \\ \hline
Bipartite                                              & \makecell{Fixed $c \geq 2$}   & \makecell{\NPc} & Theorem \ref{theo:npcomplete} \\ \hline
Chordal                                                & \makecell{Input }               & \makecell{{\XP} and {\NPc}} & \makecell{Theorems \ref{teo:c-disc-npc-chordal} \\ and \ref{teo:c-matching-xp}} \\ \hline

Interval                                                & \makecell{Input}               & $\bigO{|V|^2c\max\{|V|c, |E|\sqrt{|V|}\}}$ & \makecell{Theorem \ref{teo:c-disc-interval}} \\ \hline

Treewidth $t$                                       & \makecell{Input}               & $\bigO{8^t\eta_{t+1}^3|V|^2}$ & Theorem \ref{teo:c-disc-tw} \\ \hline

%\makecell{Any for which \\ {\sc Induced Matching} \\ is NP-complete}                                       & \makecell{Input}               & \makecell{\NPc} & Theorem \ref{teo:c-disc-induced-np-complete} \\
%\hline

\end{tabular}
\end{center}
\caption{Complexity results for { \sc Disconnected Matching } restricted to some input scopes. We denote by $\eta_i$ the $i$-th Bell number.} \label{tab:disc-p-matching}
\end{table}

%In this paper, we propose a new kind of $\mathscr{P}$-matching, the \emph{$c$-disconnected matching}, a generalization of the well known \emph{induced matching}~\cite{CAMERON198997} and the disconnected matching proposed in \cite{GODDARD2005129}. For this new matching, the induced subgraph $G[M]$ must have at least $c$ connected components. We also study and reference induced and \emph{connected matchings}, in which $G[M]$ is, respectively, $1$-regular and connected.

\noindent \textbf{Preliminaries}. For an integer $k$, we define $[k] = \{1, \dots, k\}$. For a set $S$, we say that $A,B \subseteq S$ partition $S$ if $A \cap B = \emptyset$ and $A \cup B = S$; we denote a partition of $S$ in $A$ and $B$ by $A \dot{\cup} B = S$.
A parameterized problem $\Gamma$ is said to be {\XP} when parameterized by $k$ if it admits an algorithm running in $f(k)n^{g(k)}$ time for computable functions $f,g$; it is said to be {\FPT} when parameterized by $k$ if $g \in \bigO{1}$.
We say that an \NPH\ problem $\Pi$ OR-cross-composes into a parameterized problem $\Gamma$ if, given $t$ instances $\{P_1, \dots, P_t\}$ of $\Pi$, we can build, in time polynomial on $N = \sum_{i \in [t]} |P_i|$, an instance $(x, k)$ of $\Gamma$ such that: (i) $k \leq \poly(\max\{|P_i| \mid i \in [t]\}\log t)$ and (ii) $(x, k)$ admits a solution if and only if at least one $P_i$ admits a solution.
For more on parameterized complexity, we refer to~\cite{cygan_parameterized}.
We use standard graph theory notation and nomenclature as in~\cite{murty,classes_survey}.
Let $G = (V, E)$ be a graph, $W \subseteq V(G)$, $M \subseteq E(G)$, and $V(M)$ to be the set of endpoints of edges of $M$, which are also called $M$-saturated vertices.
We denote by $G[W]$ the subgraph of $G$ induced by $W$; in an abuse of notation, we define $G[M] = G[V(M)]$.
A matching is said to be maximum if no other matching of $G$ has more edges than $M$, and perfect if $V(M) = V(G)$.
Also, $M$ is said to be connected if $G[M]$ is connected and $c$-disconnected if $G[M]$ has at least $c$ connected components.
A graph $G$ is $H$-free if $G$ has no copy of $H$ as an induced subgraph; $G$ is chordal if it has no induced cycle with more than three edges.
A graph is an interval graph if it is the intersection graph of intervals on a line.
In $G$, we denote by $\beta(G)$ the number of edges in a maximum matching, by $\beta_c(G)$ the cardinality of a maximum connected matching, by $\beta_*(G)$ the size of a maximum induced matching, and by $\beta_{d,i}(G)$ the size of a maximum $i$-disconnected matching.
If $G$ is connected, note that:
\begin{enumerate}
    \item Every maximum induced matching $M^*$ is a $\beta_*(G)$-disconnected matching, since each connected component of $G[M^*]$ is an edge.
    \item Since $\beta_*(G)$ is the maximum number of components that $G[M]$ can have with any matching $M$, there exists no $c$-disconnected matching for $c > \beta_*(G)$.
    \item Every matching is a $1$-disconnected matching.
    \item As shown in~\cite{GODDARD2005129}, $\beta(G) = \beta_c(G)$.
\end{enumerate}
Consequently, we have that both Theorem~\ref{teo:c-disc-induced-np-complete} and the following bounds hold:

\begin{equation*}
    \beta = \beta_c = \beta_{d,1} \geq \beta_{d,2} \geq \beta_{d,3} \geq \ldots \geq \beta_{d,\beta_*} \geq \beta_*
\end{equation*}

\begin{theorem}\label{teo:c-disc-induced-np-complete}
    \pname{Disconnected Matching} is {\NPc} for every graph class for which the \pname{Induced Matching} is \NPc.
\end{theorem}
\begin{proof}
Note that for every input instance $(G,k)$ to \pname{Induced Matching}, we can build an equivalent instance $(G,k,k)$ of \pname{Disconnected Matching}. That is, we want to find, in the same graph $G$, a disconnected matching $M$ with at least $k$ edges and $k$ connected components.
To obtain the induced matching, it suffices to pick, for each connected component of $G[M]$, exactly one edge.
Finally, observe that an \pname{Induced Matching} on $k$ edges is also a $k$-disconnected matching with $k$ edges.
\end{proof}

This paper is organized as follows.
%In Section \ref{sec:upper-bound}, we show upper bounds for $c$-disconnected matching cardinalities in comparison to other $\mathscr{P}$-matchings.
%In Section \ref{sec:induced-disc}, we some relations between induced and $c$-disconnected matchings are presented.
In Section \ref{sec:c-disc}, we give an alternative proof to the fact that \pname{Connected Matching} is in \P\ and present an algorithm for {\sc Maximum Connected Matching}, then present a construction used to show that \pname{$c$-Disconnected Matching} is \NPc\ for every fixed $c \geq 2$ on bipartite graphs of diameter three.
In Section \ref{sec:c-disc-chordal}, we prove our final negative result, that {\sc Disconnected Matching} is {\NPc} on chordal graphs.
We show, in Section \ref{sec:minimal-separators}, that the previous proof cannot be strengthened to fixed $c$ by giving an \XP\ algorithm for \pname{Disconnected Matching} parameterized by $c$ on graphs with a polynomial number of minimal separators.
Finally, in Sections~\ref{sec:disc-interval} and \ref{sec:tw}, we present polynomial time algorithms for \pname{Disconnected Matching} in interval and bounded treewidth graphs.
We present our concluding remarks and directions for future work in Section~\ref{sec:conclusions}.

\section{Complexity of \pname{$c$-Disconnected Matching}}
\label{sec:c-disc}

\subsection{$1$-disconnected and connected matchings}\label{sec:1-disc}

We consider that the input graph has at least one edge and is connected. Otherwise, the solution is trivial or we can solve the problem independently for each connected component.
Recall that \pname{$1$-Disconnected Matching} allows its solution to have any number of connected components. Consequently, any matching with at least $k$ edges is a valid solution to an instance $(G, k, 1)$, which leads to Theorem~\ref{teo:emp-1-desc}.

\begin{theorem}\label{teo:emp-1-desc}
{\sc $1$-Disconnected Matching} is in {\P}.
\end{theorem}
\begin{proof}
Solving the {\sc $1$-Disconnected Matching} decision problem is equivalent to ask if the answer to {\sc Maximum Matching} is $\geq k$. This equivalency is true because $G[M]$ can have any number of connected component in both. Therefore, {\sc $1$-Disconnected Matching} is in {\P} and can be solved in the same complexity of {\sc Maximum Matching}, which is $\bigO{|E|\sqrt{|V|}}$ \cite{vazirani}.
\end{proof}

Note that if $M$ is a solution to an instance $(G,k)$ of \pname{Connected Matching}, then it is also a solution to the instance $(G,k,1)$ of \pname{1-Disconnected Matching}.
Our next theorem shows that the converse is also true and, using the former theorem, that \pname{Maximum Matching} and \pname{Connected Matching} are also related. 

Based on the proof of Goddard et al.~\cite{GODDARD2005129} that $\beta(G) = \beta_c(G)$, a linear algorithm can be built to find a maximum connected matching, as described in the following theorem.

%The idea of the algorithm, based on the the equality $\beta = \beta_c$\cite{GODDARD2005129}, is that, given a maximum matching $M$, it is possible to build a connected matching with the same cardinality just by swapping edges of $M$, without changing its cardinality.

\begin{theorem}
\label{teo:connected-matching-complexity}
Given a maximum matching, a maximum connected matching can be found in linear time.
\end{theorem}

\begin{proof}

Let $M$ be the maximum matching of the input graph, and $r \in V(M)$.
We begin by running a BFS search on $G[M]$, starting at $r$, in order to obtain the connected component $S$ of $G[M]$ that contains $r$.
%The algorithm works in a BFS-like manner as follows. Given a maximum matching, start with any saturated vertex $r$ and proceed reaching as many saturated vertices as possible. Note that the set $S$ of vertices reached in this procedure is exactly the connected component induced by the matching containing the initial vertex $r$. 

Note that every vertex in $N(S)$ is unsaturated, since $S$ is maximal. Also, by the maximality of $M$, for any unsaturated vertex $v$, every vertex in $N(v)$ is saturated. Hence, if a vertex $v \in N(S)$ has any neighbor $w$ outside $S$, we can change the matching by replacing the edge saturating $w$ by $vw$. Note that now we can expand $S$ by adding $v$, $w$ and possibly other vertices, proceeding again in a BFS-like way starting from $v$ and $w$. If there is a vertex $v$ in $N(S)$ that was not considered before, repeat this step.

This can be implemented in linear time, since we never need to look any neighborhood of a vertex of $S$ twice, and during the algorithm any time a vertex $v \in N(S)$ is considered, either $v$ is included in $S$ or its neighborhood is a subset of $S$ and $v$ can be ignored.
\end{proof}

Note that, to use the algorithm described in Theorem  \ref{teo:connected-matching-complexity}, it is necessary to have calculated a maximum matching before. To obtain a maximum matching, an algorithm with complexity $\bigO{|E|\sqrt{|V|}}$ is known~\cite{vazirani}. Therefore, this complexity is the same as the { \sc Maximum Connected Matching}.

%That is, the { \sc Maximum Connected Matching } problem has a polynomial solution and its complexity is restricted to the { \sc Maximum Matching} complexity. 

\begin{corollary}
\pname{Maximum Connected Matching} has the same time complexity of \pname{Maximum matching}.
\end{corollary}

Some classes, such as trees and block graphs, have linear time algorithms for the { \sc Maximum Matching}\cite{BrunoMasquio:2019:TeseMestrado}
\cite{SAVAGE1980202}. Hence, for these classes, the algorithm we describe shows that { \sc Maximum Connected Matching } can also be solved in linear time.

\subsection{$2$-disconnected matchings}\label{subsec:2-disc-np-complete}

Next, we show that \pname{$2$-Disconnected Matching} is {\NPc} for bipartite graphs with bounded diameter. Our reduction is from the {\NPH} problem \pname{One-in-three 3SAT}~\cite{garey_johnson}; in this problem, we are given a set of $m$ clauses $I$ with exactly three literals in each clause, and asked if there is a truth assignment of the variables such that only one literal of each clause resolves to true.
We consider that each variable must be present in at least one clause and that a variable is not repeated in the same clause. This follows from the original {\NPcness} reduction \cite{schaefer_1_3_sat}.

\subsubsection{Input transformation in {\sc One-in-three 3SAT}.}\label{subsec:transformacaoEntrada}

We use $k = 12m$ and build a bipartite graph $G = (V_1 \dot{\cup} V_2,E)$ from a set of clauses $I$ as follows.

\begin{enumerate}[(I)]

    \item For each clause $c_i$, generate a subgraph $B_i$ as described below.
    %Observe Figure \ref{fig:subgrafoBi} as an example of such clause subgraph.

    \begin{itemize}
        \item $V(B_i) = \{ l_{ij}, r_{ij}$ $|$ $j \in \{1,\ldots,9 \} \}$
        \item $E(B_i)$ is as shown in Figure \ref{fig:subgrafoBi}.
        
    \end{itemize}

    \begin{figure}
    \centering
    \includegraphics[width=0.45\textwidth]{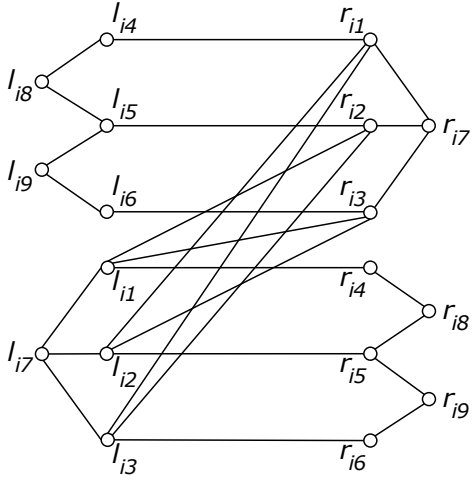}
    \caption{The subgraph $B_i$, related to clause $c_i$}
    \label{fig:subgrafoBi}
    \end{figure}

    \item For each variable $x$ present in two clauses $c_i$ and $c_j$, being the $q$-th literal of $c_i$ and the $t$-th literal of $c_j$, add two edges. If $x$ is negated in exactly one of the clauses, add the set of edges $\{r_{iq}l_{jt}, l_{i(q+3)}r_{j(t+3)} \}$. Otherwise, add $\{ l_{i(q+3)}r_{jt}, r_{iq}l_{j(t+3)} \}$.
    \begin{itemize}
        \item As an example, consider $c_2 = (x \vee y \vee \overline{z})$ and $c_5 = (d \vee \overline{x} \vee \overline{g})$. Variable $x$ is present in $c_2$ as the first literal and, in $c_5$, as the second literal. Besides, $x$ is negated only in $c_5$. Hence, we add the edge set $\{ (r_{21},l_{52}), (l_{24},r_{55}) \}$. In Figure \ref{fig:exemplo2}, we present this example, but omit some edges for better visualization.
    \end{itemize}
    
    \begin{figure}
    \centering
    \includegraphics[width=0.9\textwidth]{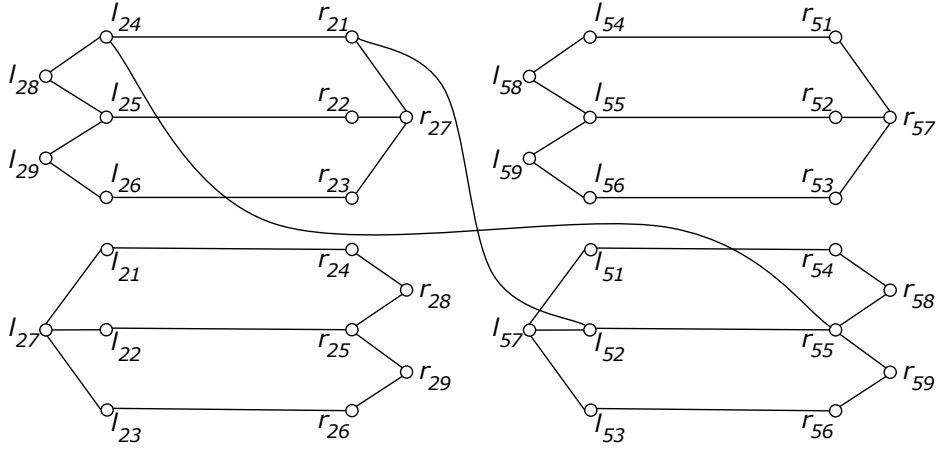}
    \caption{The simplified subgraph $G[V(B_2) \cup V(B_5)]$ for an instance with the clauses $c_2 = (x \vee y \vee \overline{z})$ e $c_5 = (d \vee \overline{x} \vee \overline{g})$}
    \label{fig:exemplo2}
    \end{figure}
    
    \item Generate two complete bipartite subgraphs $H_1$ and $H_2$, both isomorphic to $K_{3m,3m}$, $V(H_1) = V(U_1) \dot{\cup} V(U_2)$ and $V(H_2) = V(U_3) \dot{\cup} V(U_4)$.

    \item For each $u_2 \in V(U_2)$ and clause $c_i$, add the edge set $\{ u_2l_{ij} \mid j \in \{1,\ldots,6 \} \}$.

    \item For each $u_3 \in V(U_3)$ and clause $c_i$, add the edge set $\{ u_3r_{ij} \mid j \in \{1,\ldots,6 \} \}$. %Observe Figure \ref{fig:subgrafoBiComH1eH2} as an example of edges from items (IV) e (V).

\end{enumerate}

%Note that the bipartition $V_1$ of $G = (V_1 \dot{\cup} V_2, E)$ can be defined as $V_1 = \{ l_{iq} $ $|$ $i \in [m], q \in [6] \}$ $\cup$ $\{ r_{iq} \!\mid i \in [m], q \in \{7,8,9\} \}$ $\cup$ $\{ u_1, u_3 \mid u_1 \in V(U_1), u_3 \in V(U_3) \}$. Given a clause subgraph $B_i$, Figure \ref{fig:subgrafoBiComH1eH2} illustrates the subgraph $G[V(B_i) \cup \{ u_1, u_2, u_3, u_4 \}]$, $u_1 \in V(U_1)$, $u_2 \in V(U_2)$, $u_3 \in V(U_3)$ and $u_4 \in V(U_4)$ with different vertices colors for each bipartition $V_1$ and $V_2$.

Besides $G$ being bipartite, as shown in Figure \ref{fig:subgrafoBiComH1eH2}, it is possible to observe that its diameter is $5$, regardless of the set of clauses and its cardinality. This holds due to the distance between, for example, $u_1$ and $u_4$, $u_1 \in V(U_1)$, $u_4 \in V(U_4)$, as well as $l_{i7}$ and $r_{j7}$, $i,j \in [m]$, distinct, such that the clauses $i$ and $j$ do not have literals related to the same variable. Also, consider $G_i^+ =  G[V(B_i) \cup V(H_1) \cup V(H_2)]$. Note that $|V(G)| = \bigO{m}$.

 \begin{figure}
    \centering
    \includegraphics[width=0.8\textwidth]{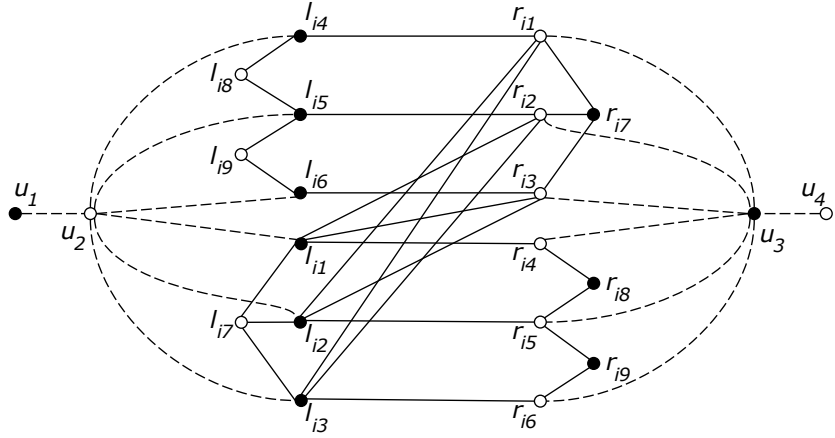}
    \caption{The subgraph $G[V(B_i) \cup \{ u_1, u_2, u_3, u_4 \}]$, $u_1 \in V(U_1)$, $u_2 \in V(U_2)$, $u_3 \in V(U_3)$ and $u_4 \in V(U_4)$. The bold vertices represent a bipartition of $G$.}
    \label{fig:subgrafoBiComH1eH2}
    \end{figure}

We denote $A$ by the subgraph induced by the vertices of $V(G) \setminus V(H_1) \setminus V(H_2)$. Note that $A$ is exactly the subgraph induced by the vertex set of $B_i$, $i \in [m]$. 

Observe that $|V(G)| = 30m$, since $V(G) = V(H_1) \cup V(H_2) \cup V(A)$. Besides, $(45m^2 + 26m) \leq |E(G)| \leq (48m^2 + 23m)$, as the following amounts of edges are generated in the construction. There are $26m$ in (I), $9m^2 $ in (III), $36m^2$ in (IV) and (V). In (II), we can have from $0$ to $3m^2 - 3m$, since each pair of clause subgraphs can have from $0$ to $6$ edges between each other.

%and $|E(G)| = \Omega(m^2)$.

%Note that $|V(G)| = 30m$, since $V(G) = V(H_1) \cup V(H_2) \cup V(A)$. Besides, $(45m^2 + 26m) \leq |E(G)| \leq (48m^2 + 23m)$, as the following amounts of edges are generated in the construction. We have $26m$ in (I), $9m^2 $ in (III), $36m^2$ in (IV) and (V). In (II), we can have from $0$ to $3m^2 - 3m$, since each pair of clause subgraphs can have from $0$ to $6$ edges between each other.

\subsubsection{Properties of disconnected matchings in the generated graphs.}\label{subsec:propriedadesEmpDesconexos}

We now prove some properties of the disconnected matching with cardinality at least $k$ in a graph $G$ generated by the transformation described.

Initially, we show, from Lemmas \ref{lemma:H1H2matched} and \ref{lemma:duasComponentesConexas}, that a subgraph induced by the saturated vertices of such matching has exactly two connected components, one containing vertices of $H_1$ and the other, vertices of $H_2$. Afterwards, Lemma \ref{lemma:6-arestas} shows the possible sets of edges contained in the matching. 

\begin{lemma}\label{lemma:H1H2matched}
If $M$ is a disconnected matching with cardinality $k \geq 12m$, then there exists two saturated vertices $h_1 \in V(H_1)$ and $h_2 \in V(H_2)$.
\end{lemma}
\begin{proof}
In order to obtain $M$ with cardinality $k$, it is necessary that $2k \geq 2\cdot(12m) = 24m$ vertices are saturated by $M$. Note that $|V(A)| = 18m$. Since we are looking for a $k$ cardinality matching, then, even if all the vertices of $A$ were saturated, we would have, at most, $18m$ vertices. Therefore, for $M$ to saturate $2k$ vertices, we need to use vertices of $V(H_1) \cup V(H_2)$. Note that, similarly, if vertices of $A$ and only one of the subgraphs $H_1$ or $H_2$, $|V(H_1)| = |V(H_2)| = 6m$, we will have a maximum of $24m$ vertices, however, the matching would be perfect in the subgraph and, therefore, connected. Thus, it is necessary that there are at least two vertices $h_1 \in V(H_1)$ and $h_2 \in V(H_2)$ saturated by $M$.
\end{proof}

\begin{lemma}\label{lemma:duasComponentesConexas}
If $M$ is a disconnected matching with cardinality $k \geq 12m$, then $G[M]$ has exactly two connected components.
\end{lemma}
\begin{proof}
From Lemma \ref{lemma:H1H2matched}, we know that $M$ saturates $h_1 \in V(H_1)$ and $h_2 \in V(H_2)$ by two edges, $(h_1, v_1)$ and $(h_2, v_2)$. Note that, due to the graph structure, every edge $e$ saturated by $M$ is incident to at least one vertex of $V(G) \setminus \{ l_{i7}, l_{i8}, l_{i9}, r_{i7}, r_{i8}, r_{i9} \}$, $i \in [m] $, as every edge of the graph has this property. Then, one end of $e$ is adjacent to any of the vertices in $\{ h_1, v_1, h_2, v_2 \}$. Therefore, $G[M]$ has exactly two connected components $C_1$ and $C_2$ such that $h_1 \in V(C_1)$ and $h_2 \in V(C_2)$.
\end{proof}

\begin{lemma}\label{lemma:6-arestas}
Let $M$ be a disconnected matching with cardinality $k \geq 12m$ and $B_i$ be a clause subgraph. There are exactly $6$ edges saturated by $M$ in $G[V(B_i)]$ and, moreover, there are exactly $3$ sets of edges that satisfy this constraint.
\end{lemma}
\begin{proof}
Let $M$ be a $2$-disconnected matching in $G$, $|M| \geq 12m$, and $B_i$ be a clause subgraph. From Lemma \ref{lemma:H1H2matched}, we know that there are two saturated vertices $u_2 \in V(U_1)$ and $u_3 \in V(U_3)$. Also, Lemma \ref{lemma:duasComponentesConexas} shows that any other saturated vertex in the graph must be in one of the two connected components of $G[M]$ containing $u_2$ or $u_3$. Therefore, there is a $u_2-u_3$ separator $S_i$ not saturated in $B_i$.

For the rest of the proof, we use $S_i$ separators with cardinality $6$, so that all the $12$ vertices of $V(B_i) \setminus S_i$ will be saturated by $M$. We prove that there are only $3$ separators of this type, due to the following properties.

\begin{enumerate}
    \item The vertex pairs $l_{ij}$ and $r_{i(j+3)}$ cannot be saturated simultaneously. Thus, $S_i$ contains at least one of these two vertices.
    \item If there are two saturated vertices $l_{ij}$ and $l_{iq}$, then the four vertices $r_{ij}$, $r_{iq}$, $r_{it}$ and $r_{i7}$ cannot be saturated, $j,q,t \in \{ 1,2,3 \}$, distinct.
    \item The vertex $l_{ij}$ cannot be saturated simultaneously with $r_{iq}$ or $r_{it}$, $j,q,t \in \{ 1,2,3 \}$, distinct.
    \item The number of saturated vertices of $\{ l_{i8}, l_{i9} \}$ must be at most the number of saturated vertices of $\{ l_{i4},l_{i5},l_{i6} \}$.
\end{enumerate}

Consider $W_i$ the set of vertices in $V(B_i) \setminus S_i$ that could possibly be saturated by $M$. From Property $1$, we can see that $| S_i | \geq 6$. Thereby, $| W_i | \leq 12 $, that is, the largest number of saturated vertices in a clause subgraph is $12$. Next, we'll show that there are only $3$ separators and $W_i$ sets of that type. If $|S_i| = 6$, which is its minimum cardinalty, then, given Property $2$, there can only be a single saturated vertex $l_{ij}$, $j \in \{ 1,2,3 \}$. In addition, for $r_{i7}$ to be in $W_i$, a vertex $r_{it}$, $t \in \{ 1,2,3 \}$ must be in $W_i$ as well. Given Property $3$, the only possibility is if $t = j$. Thus, the vertices $l_{i (j + 3)} $ and $r_{i (j + 3)}$ cannot belong to $W_i$. In addition, from Property $4$, for $l_{i8}$ and $l_{i9}$ to be in $W_i$, then two vertices $l_{i (q + 3)}$ are in $W_i$ as well, $q \in \{ 1,2,3 \}$. The only possibility of this occurring is if $q \neq j$. Analogously, the same is true for the vertices $r_{i (q + 3)}$. Finally, we can define the set described as $W_i =$ $\{ l_{ij}, r_{ij}, l_{i(q+3)}, r_{i(q+3)}, l_{i(t+3)}, r_{i(t+3)}$ $|$ $ j,q,t \in \{ 1,2,3 \}, \text{distinct} \}$ $\cup$ $\{ l_{ij}, r_{ij} \mid j \in \{ 7,8,9 \} \}$. Therefore, there are only $3$ possibilities for the set $W_i$, which are shown in Table \ref{tab:separadoresBi} and in Figure \ref{fig:separadoresBi}. Moreover, there is exactly one corresponding saturated edge set for each of the vertex sets, shown in Figure \ref{fig:separadoresBi}.

\begin{figure}%
        \centering
        \subfloat[]{{\includegraphics[width=0.5\textwidth]{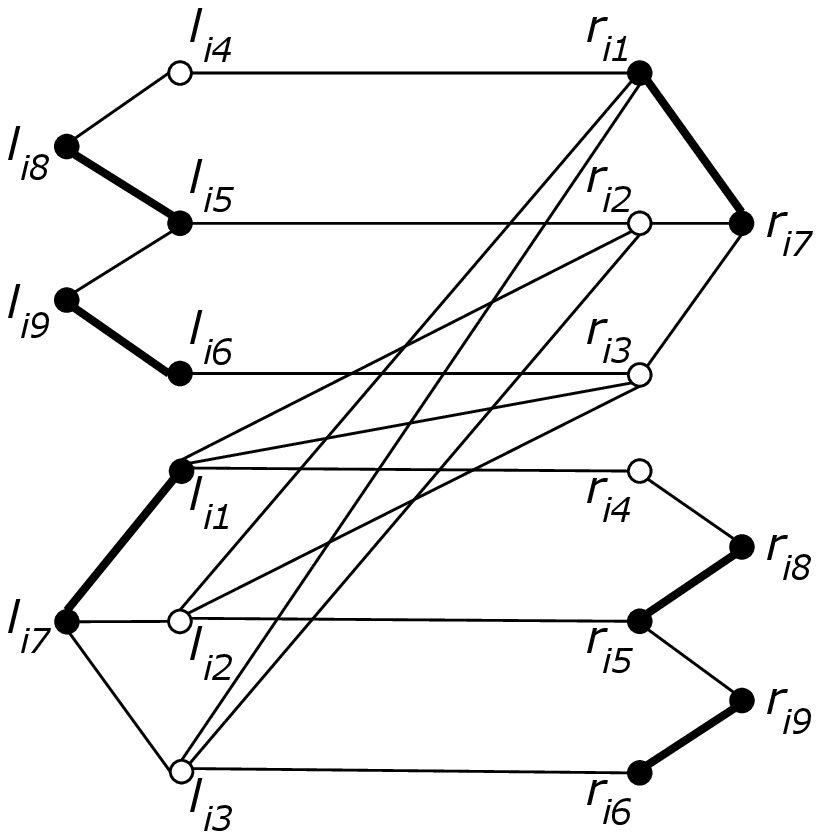} }}
        \subfloat[]{{\includegraphics[width=0.5\textwidth]{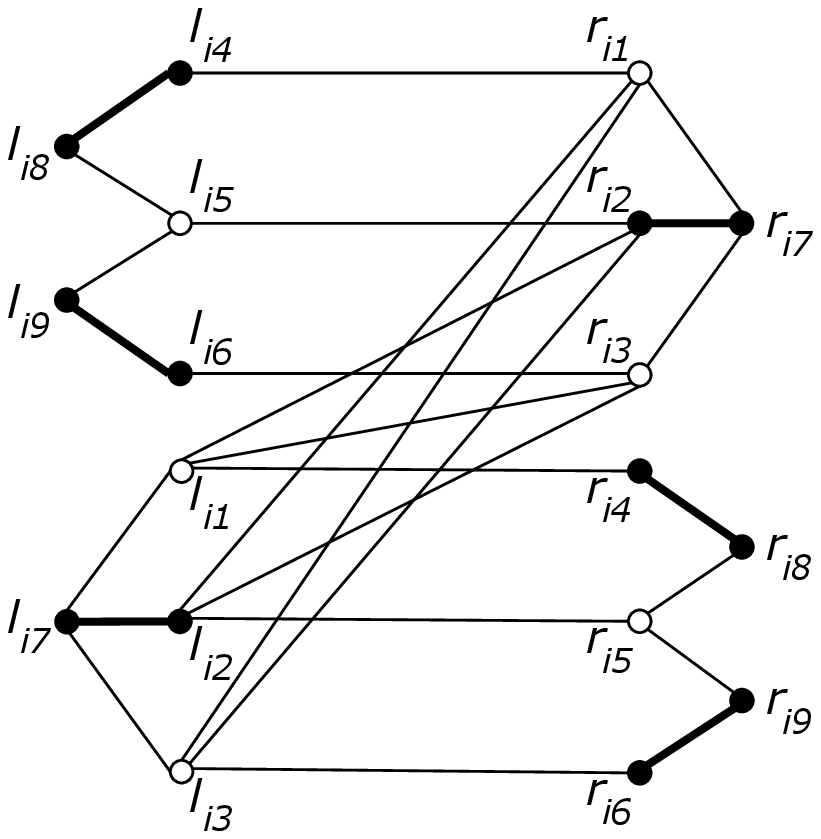} }}
        
        \subfloat[]{{\includegraphics[width=0.5\textwidth]{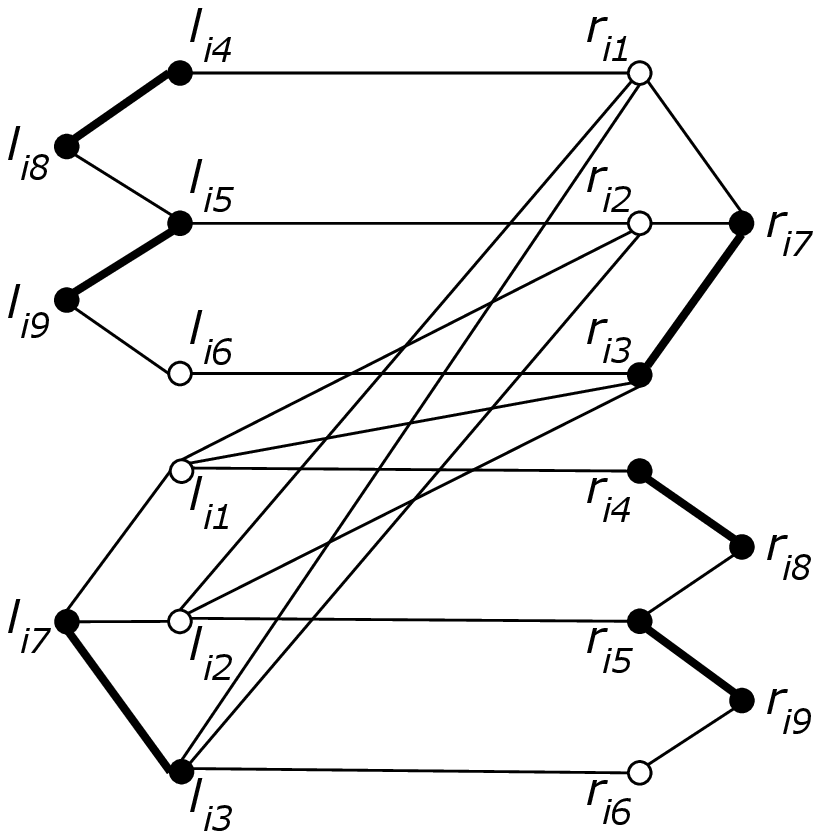} }}
        \caption{Subgraph $B_i$ with, in bold, the three sets of vertices possibly saturated with cardinality $12$ from Table \ref{tab:separadoresBi}.}%
        \label{fig:separadoresBi}
\end{figure}

\end{proof}

\subsubsection{Transforming a disconnected matching into a variable assignment.}\label{subsubsec:emparelhamentoDesconexoemAtribuicao}

First, we define, starting from a $2$-disconnected matching $M$, $|M| = 12m$, a variable assignment $R$ and, in sequence, we present Lemma \ref{lemma:NPHardVolta}, proving that $R$ is a {\sc One-in-three 3SAT} solution.

\begin{enumerate}[(I)]
    \item For each clause $c_i$, where $x_{ij}$ corresponds to the $j$-th literal of $c_i$, generate the following assignments.
    
    \begin{itemize}
        \item If $l_{ij}$ is $M$-saturated, then assign $x_{ij} = T$.
        \item Otherwise, assign $x_{ij} = F$.
    \end{itemize}
\end{enumerate}

\begin{comment}
\begin{enumerate}[(I)]
    \item Para cada cláusula $c_i$, em que $x_{ij}$ corresponde ao $j$-ésimo literal de $c_i$, gerar as seguintes atribuições.
    
    \begin{itemize}
        \item Se os vértices $l_{ij}$ e $r_{ij}$ estão saturados por $M$, então $x_{ij} = V$.
        \item Senão, $x_{ij} = F$.
    \end{itemize}
\end{enumerate}
\end{comment}

Note that, analyzing the generated graph, the pair of saturated vertices $l_{ij}$ and $r_{ij} $, $j \in \{1,2,3 \}$ define that the $j$-th literal is the true of the clause $c_i$. Similarly, each pair of saturated vertices $l_{iq}$ and $r_{iq}$, $q \in \{4,5,6 \}$, $q \neq j + 3$, defines that the $(q-3)$-th literal is false.

\begin{lemma}\label{lemma:NPHardVolta}
Let $M$ be a $2$-disconnected matching with cardinality $k = 12m$ in a graph generated from a input $I$ of {\sc One-in-three 3SAT}. It is possible to generate in polynomial time an assignment to variables in $I$ that solves {\sc One-in-three 3SAT}.
\end{lemma}
\begin{proof}

For this Lemma to be true, using the $R$ assignments, each clause in $I$ must have exactly one true literal and each variable must have the same assignment in all clauses. As was deduced in the Lemma \ref{lemma:6-arestas}, in fact, given $i$, $l_ {ij}$ is saturated only for a single $j$, $j \in \{ 1,2,3 \}$, $i \in [m]$. So we have a single true literal. We now show that the assignment of the variable is consistent across all clauses. By contradiction, assume this to be false, then in $R$ there are two literals, $x$ and $y$, for the same variable, and one of the following two possibilities occurs. Consider $x$ the $q$-th literal of $c_i$ and $y$ the $t$-th literal of $c_j$. Either there is one negation between $x$ and $y$ or $x$ and $y$ have the same sign. In the first possibility, as we assumed that $x$ and $y$ have different assignments, then either $r_{iq}$ and $l_{jt}$ are saturated simultaneously or $ l_{i (q + 3)}$ and $r_{j (t + 3)} $ are. Note that $c_i$ and $c_j$ have variables with opposite literals, which means that the constructed graph has the edges $r_{iq}l_{jt}$ and $l_{i(q + 3)}r_{j(t+3)}$. Therefore, by the Lemma \ref{lemma:H1H2matched}, $G[M]$ would be connected, which is a contradiction. In the second possibility, $x$ and $y$ have the same sign. So either $l_{i (q + 3)}$ and $r_{jt}$ are saturated simultaneously or $r_{iq}$ and $l_{j (t + 3)}$ are. There are also edges between these pairs of vertices and, also by Lemma \ref{lemma:H1H2matched}, it is a contradiction. Therefore, $R$ solves {\sc One-in-three 3SAT}.
\end{proof}

\begin{table}
\begin{center}
\begin{tabular}{ |c|c| } 
 \hline
  $h_1-h_2$ separator of $G_i^+$ & Possibly saturated remaining vertices \\ 
 \hline
    $\{ l_{i2}, l_{i3}, l_{i4}, r_{i2}, r_{i3}, r_{i4} \}$ & $\{ l_{i1}, l_{i5}, l_{i6}, l_{i7}, l_{i8}, l_{i9}, r_{i1}, r_{i5}, r_{i6}, r_{i7}, r_{i8}, r_{i9}\}$ \\
    
    $\{ l_{i1}, l_{i3}, l_{i5}, r_{i1}, r_{i3}, r_{i5} \}$ & $\{ l_{i2}, l_{i4}, l_{i6}, l_{i7}, l_{i8}, l_{i9}, r_{i2}, r_{i4}, r_{i6}, r_{i7}, r_{i8}, r_{i9}\}$ \\
    
    $\{ l_{i1}, l_{i2}, l_{i6}, r_{i1}, r_{i2}, r_{i6} \}$ & $\{ l_{i3}, l_{i4}, l_{i5}, l_{i7}, l_{i8}, l_{i9}, r_{i3}, r_{i4}, r_{i5}, r_{i7}, r_{i8}, r_{i9}\}$ \\
 \hline
\end{tabular}
\end{center}
\caption{Some $h_1-h_2$ minimal separators in a subgraph $G_i^+$ and the respective remaining sets of vertices, which can be saturated} \label{tab:separadoresBi}
\end{table}

%\todo{Tirei frases aqui}
%Analyzing the subgraph $G_i^+$ and the sets of vertices that can be saturated by the Lemma \ref{lemma:6-arestas}, we can define the edges that are used.

%Note that the only possible edge sets can be generated from the following process. For each subgraph $B_i$, generate a matching with cardinality $6$ where each edge is disjointly incident to one of the vertices $\{ l_{i7}, l_{i8}, l_{i9}, r_{i7}, r_{i8}, r_{i9} \}$. Since $H_1$ and $H_2$ are complete bipartite subgraphs with the same number of vertices for each partition, it is possible to generate a matching by saturating all vertices, with $6m$ edges. At the end of the process, we have the $12m$ edges.

%Therefore, the $12$ saturated vertices of $B_i$ indicated in Lemma \ref{lemma:6-arestas} corresponds to $6$ edges in $G[V(B_i)]$.

\subsubsection{Transforming a variable assignment into a disconnected matching.}\label{subsubsec:atribuicaoParaEmparelhamentoDesconexo}

Finally, we define a $2$-disconnected matching $M$, obtained from a solution of {\sc One-in-three 3SAT}. Then, Lemma \ref{lemma:NPHardIda} proves that $M$ is a $2$-disconnected matching with the desired cardinality $12m$.

%Finally, we define a disconnected matching $M$ in the transformation bipartite graph $G$ described below, $|M| = 12m$, from an assignment of variables $R$, which is a solution of the input $I$ to the problem {\sc One-in-three 3SAT}. Then, Lemma \ref{lemma:NPHardIda} proves that $M$ is a disconnected matching with the desired cardinality.

\begin{enumerate}[(I)]
    \item For each clause $c_i$, whose true literal is the $j$-th, add to $M$ the edge set defined as $\{ l_{ij}l_{i7}, r_{ij}r_{i7}, l_{iq}l_{i8}, r_{iq}r_{i8}, r_{it}l_{i9}, r_{it}r_{i9} $ $|$ $ q \in \{ 4,5 \}, t \in \{ 5,6 \}, q \neq j+3 \neq t \neq q \}$.
    
    \item For $H_1$, add to the matching $M$ any $3m$ disjoint edges. Repeat the process for $H_2$.
\end{enumerate}

\begin{comment}
\begin{enumerate}[(I)]
    \item Para cada cláusula $c_i$, cujo literal verdadeiro é o $j$-ésimo, adicionar a $M$ o conjunto de arestas definido por $\{ (l_{ij}, l_{i7}), (r_{ij}, r_{i7}), (l_{iq}, l_{i8}), (r_{iq}, r_{i8}), (r_{it}, l_{i9}), (r_{it}, r_{i9}) $ $|$ $ q \in \{ 4,5 \}, t \in \{ 5,6 \}, q \neq j+3 \neq t \neq q \}$.
    
    \item Para $H_1$, adicione ao emparelhamento $M$ as $3m$ arestas, cada uma incidente de forma disjunta a quaisquer dois vértices, um de cada bipartição de $H_1$. Repita o processo para $H_2$.
\end{enumerate}
\end{comment}

\begin{lemma}\label{lemma:NPHardIda}
Let $R$ be a variable assignment of an input $I$ from {\sc One-in-three 3SAT}. It is possible, in polynomial time, to generate a disconnected matching with cardinality $k = 12m$ from $I$ in a graph generated by the transformation described below.
\end{lemma}
\begin{proof}
In the procedure described, we are saturating $6$ edges for each clause in (I), and $6m$ edges in (II). Then, $M$ has $12m$ edges. We need to show now that $M$ is disconnected. It is necessary and sufficient showing that there are not two adjacent vertices $l$ and $r$ saturated. Edges between vertices of the same clause subgraph are generated in (I) and we observe that there are no two adjacent saturated vertices of this type, since the saturated vertices are those described in Lemma \ref{lemma:6-arestas}. The vertices incident to the edges between different clause subgraphs cannot be simultaneously saturated, as they would represent variables and their negations as true. Therefore, $M$ is disconnected and $|M| = 12m$.
\end{proof}

Note that for any graph with diameter $d \leq 1$ the answer to \pname{Disconnected Matching} is always \NOi.
On the other hand, if the graph is disconnected, there are two possibilities.
If the graph has no more than one connected component with more than one vertex, we again answer \NOi.
Otherwise, the problem can be solved in polynomial time by finding a maximum matching $M$ and checking if $|M| \geq k$.
These statements are used in the proof of Lemma \ref{lemma:disc-match-np-c-diameter}, which has a slight modification of the above construction, but allows us to reduce the diameter of the graph to $3$.

\begin{lemma}\label{lemma:disc-match-np-c-diameter}
Let $G = (V_1 \dot{\cup} V_2, E)$ be the bipartite graph from the transformation mentioned and $G'= (V_1' \dot{\cup} V_2', E')$ so that $V(G') = V(G) \cup \{w_1, w_2 \} $ and
$E(G') = E(G) \cup \{w_1w_2\} \cup \{vw_1 \mid v \in V (V_1) \} \cup \{vw_2 \mid v \in V (V_2) \}$.
If $M$ is a $2$-disconnected matching in $G'$, $|M| \geq k$, so $M$ is also a $2$-disconnected matching in $G$.
\end{lemma}
\begin{proof}
Let's show that a $2$-disconnected matching $M$ in $G'$, $|M| \geq 12m$, saturates only vertices of $V(G)$ and, therefore, $M$ is also a $2$-disconnected matching in $G$. With this purpose, we demonstrate that the vertices $w_1$ and $w_2$ are not part of $M$. Let's assume that $w_1$ is saturated and $w_1 \in C_1$. Note that, since $G'$ is bipartite, then every edge $e \in E(G')$ has one endpoint at $V(V_1')$ and other at $V(V_2')$. Therefore, the edge $e$, if saturated, would be at $C_1$. 
Thereby, $M$ would not be $2$-disconnected, which is a contradiction. This shows that $w_1$ is not saturated. The argument is analogous to $w_2$. Thus, if $M$ is a $2$-disconnected matching in $G'$, $|M| \geq 12m$, so it's also in $G$. As we have already described the structure of such matchings in $G$ in Lemmas \ref{lemma:H1H2matched}, \ref{lemma:duasComponentesConexas} and \ref{lemma:6-arestas}, this transformation can also be used to solve the {\sc One-in-three 3SAT} problem.
\end{proof}

Combining the previous results, we obtain Theorem \ref{theo:npcomplete}.

\begin{theorem}\label{theo:npcomplete}
{\sc $2$-Disconnected Mathing} is {\NPc} even if the input is restricted to bipartite graphs with diameter $3$.
\end{theorem}
\begin{proof}
Let $G = (V_1 \cup V_2, E)$ be a graph generated from the transformation of Section \ref{subsec:transformacaoEntrada}. Let's show that this graph is bipartite. Note that the bipartition $V_1$ of $G$ can be defined by $V_1 = \{ l_{iq} $ $|$ $i \in [m], q \in \{ 1,\ldots,6 \} \}$ $\cup$ $ \{ r_{iq} $ $|$ $i \in [m], q \in \{7,8,9\} \}$ $\cup$ $\{ u_1, u_3 $ $|$ $u_1 \in V(U_1), u_3 \in V(U_3) \}$.

Next, we prove that the problem is in {\NP} and {\NPH}. Note that a $2$-disconnected matching is a certificate to show that the problem is in {\NP}. According to the transformations between {\sc $2$-Disconnected Matching} and {\sc One-in-Three 3SAT} solutions described in Lemmas \ref{lemma:NPHardIda} and \ref{lemma:NPHardVolta}, the {\sc One-in-three 3SAT} problem, which is {\NPc}, can be reduced to {\sc $2$-Disconnected Matching} using a diameter $3$ bipartite graph. Therefore, {\sc $2$-Disconnected Matching} is {\NPH} and we have proven that {\sc $2$-Disconnected Matching} is {\NPc} even for diameter $3$ bipartite graphs.
\end{proof}

These results imply the following dichotomies, in terms of diameter.

\begin{corollary}\label{coro:dicothomyBipartite}
For bipartite graphs with diameter $ \leq d$, {\sc Disconnected Matching} is {\NPc} if $d$ is at least $3$ and belongs to $P$ otherwise.
\end{corollary}

\begin{corollary}\label{coro:dicothomyGeneral}
For graphs with diameter $\leq d$, {\sc Disconnected Matching} is {\NPc} if $d$ is at least $2$ and belongs to $P$ otherwise.
\end{corollary}

\subsubsection{Example of {\sc $2$-Disconnected Matching} transformation}\label{subsec:example}

Consider the input $I$ with the two clauses $c_1 = (x \vee y \vee z)$ and $c_2 = (w \vee y \vee \overline{x})$. 

Let's build the graph $G$, $V(G) = 30m = 60$, $E(G) = 47m^2 + 24m = 236$, as described in Section \ref{subsec:transformacaoEntrada}.

We describe below the edges between $B_1$ and $B_2$. Note that in $I$, the third literal of $c_2$ is the negation of $c_1$. Therefore, the edges $\{ (r_{11}, l_{23}), (l_{14}, r_{26}) \}$ must be added. In addition, the second literal of $c_1$ is the same as $c_2$. Thus, we add the edges $\{ (l_{15}, r_{22}), (r_{12}, l_{25})) \}$. The rest of the literals refer to different variables, so there are no additional edges between $B_1$ and $B_2$.

We present the only two $I$ solutions for \pname{One-in-three 3SAT} and their corresponding disconnected matchings in $G$. The assignment of the variables $(w, x, y, z)$, in this order, can be either $(V, V, F, F)$, represented by the matching in Figure \ref{fig:exemploEmparelhamento1}, or $( F, F, F, V)$, in Figure \ref{fig:exemploEmparelhamento2}. For easier visualization, some edges of $B_1$ and $B_2$ are omitted, besides the complete subgraphs $H_1$ and $H_2$ and their respective $6m = 12$ saturated edges.

\begin{figure}%
        \centering
        \subfloat[]{{\includegraphics[width=0.5\textwidth]{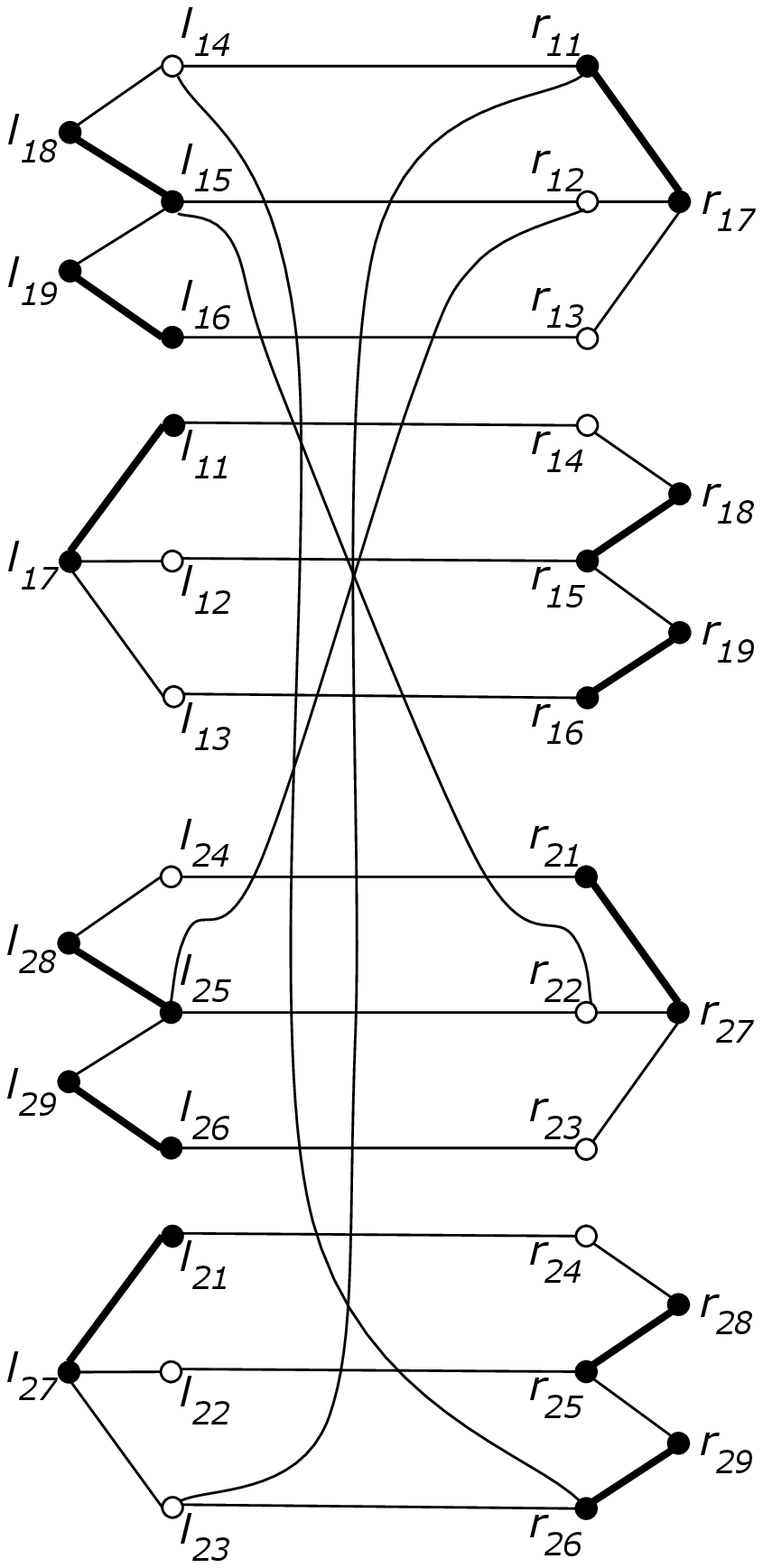}\label{fig:exemploEmparelhamento1} }}
        \subfloat[]{{\includegraphics[width=0.5\textwidth]{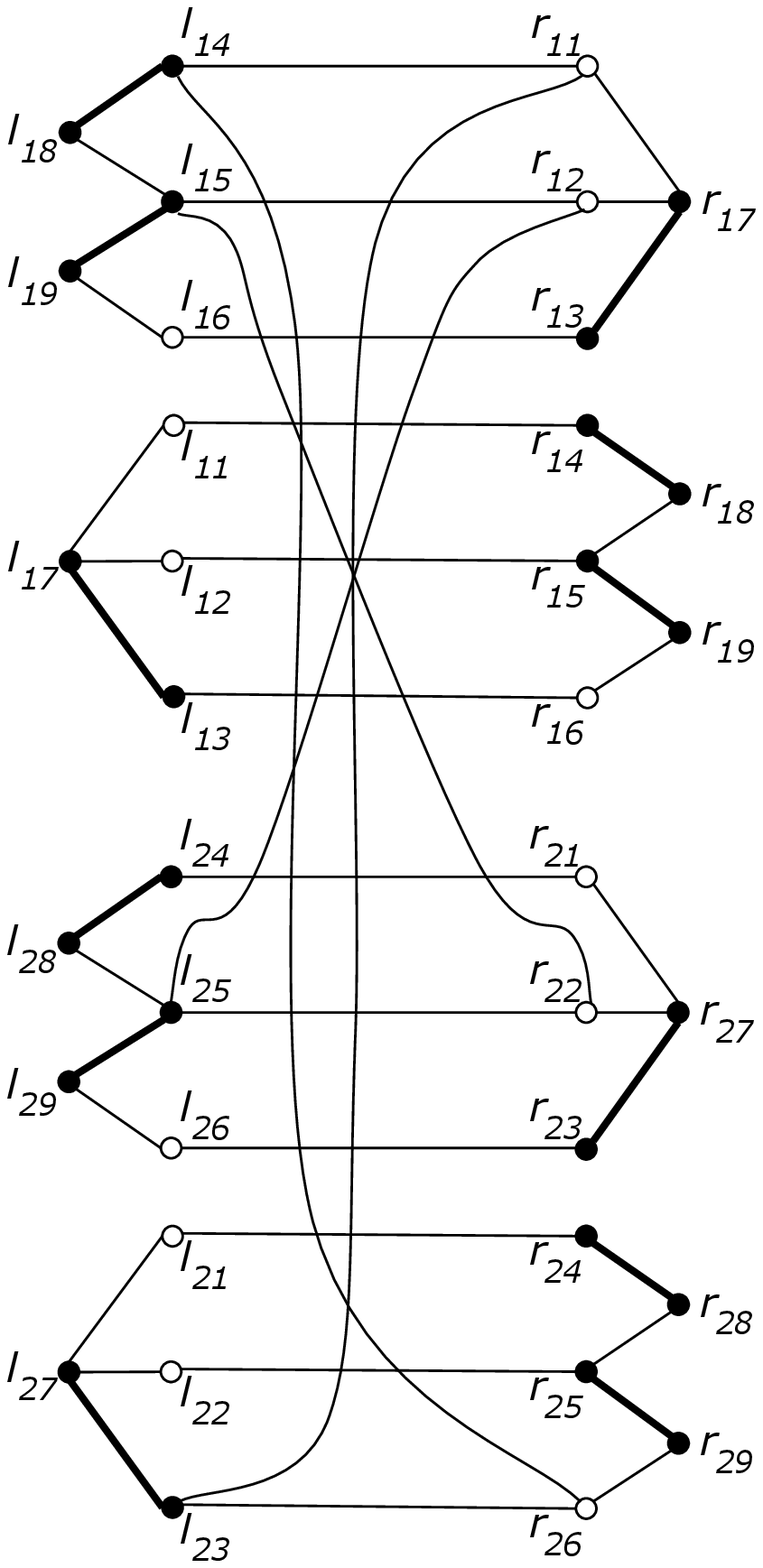}\label{fig:exemploEmparelhamento2} }}
        \caption{Disconnected matching examples in the simplified clause subgraphs for the boolean expression $(x \vee y \vee z) \wedge (w \vee y \vee \overline{x})$}.%
        \label{fig:exemploEmparelhamento}
\end{figure}

\subsection{NP-completeness for any fixed $c$}\label{subsec:c-disc-np-complete}

We now generalize our hardness proof to \pname{$c$-Disconnected Matching} for every fixed $c > 2$.
We begin by setting the number of edges in the matching $k = 12m + c - 2$, defining $G'$ to be the graph obtained in our hardness proof for \pname{$2$-Disconnected Matching}, and $H$ to be the graph with $c - 2$ isolated edges $\{v_{i1}v_{i2} \mid i \in [c-2]\}$.
To obtain our input graph to \pname{$c$-Disconnected Matching}, we make $w_1 \in V(G')$ adjacent to $v_{i1}$ and $w_2 \in V(G')$ adjacent $v_{i2}$, for every $i \in [c-2]$, where $w_1$ and $w_2$ are as defined in Lemma~\ref{lemma:disc-match-np-c-diameter}.
This proves that the problem is \NPc\ on bipartite graphs of diameter three.
Note that, if we identify $w_1$ and $w_2$, we may reason as before, but now conclude that \pname{$c$-Disconnected Matching} is \NPc\ on general graphs of diameter 2. We summarize this discussion as Theorem~\ref{teo:dicotomia-emp-c-desc}.

\begin{theorem}
\label{teo:dicotomia-emp-c-desc}
Let $c \geq 1$. The { \sc $c$-Disconnected Matching} problem belongs to \P\ if $c = 1$. Otherwise, it is {\NPc} even for bipartite graphs of diameter $3$ or for general graphs of diameter $2$.
\end{theorem}
\begin{proof}
It is simple to verify that the problem belongs to \NP, since a certificate can be a matching that induces a graph with $c$ connected components. Such matching can be verified in polynomial time. We show next that the problem is {\NPH} for some cases and {\P} for others. For $c>1$, the {\sc One-in-three 3SAT} can be reduced to {\sc $c$-Disconnected Matching} using, as an input, the transformation graph and $k = 12m + (c-2)$. Let $M$ be a matching such that $G[M]$ has $c$ connected components and $|M| \geq k$. Let's consider two $M$ partitions, contained in $U$ and $G'$. Regarding the subgraph $U$, if we consider that all edges of $E(U)$ are saturated, then $G[M]$ will have at least $c-2$ connected components. Regarding $G'$, we know that for $G'[M]$ to have the remaining $2$ connected components, the largest number of saturated edges in $G'$ must be $12m$. In total, $M$ will have exactly $12m + c - 2$ edges. Concerning the restriction on the number of connected components of $G[M]$, we conclude that $M$ is maximum and all maximum $c$-matchings in $G$ have this form. For $c = 1$, Theorem \ref{teo:emp-1-desc} shows that the problem belongs to \P.
\end{proof}

\section{NP-completeness for chordal graphs}\label{sec:c-disc-chordal}

In this section, we prove that {\sc Disconnected Matching} is {\NPc} even for chordal graphs with diameter $2$. In order to prove it, we describe a reduction from the {\NPc} problem { \sc Exact Cover By $3$-Sets } \cite{garey_johnson}. This problem consists in, given two sets $X$, $|X| = 3q$, and $C$, $|C| = m$ of $3$-element subsets of $X$, decide if there exists a subset $C' \subseteq C$ such that every element of $X$ occurs in exactly one member of $C'$.

For the reduction, we define $c = m-q+1$, $k = m + 3q$ and build the chordal graph $G = (V,E)$ from the sets $C$ and $X$ as follows.

\begin{enumerate}[(I)]
    \item For each $3$-element set $c_i = (x,y,z)$, $c_i \in C$, generate a complete subgraph $H_i$ isomorphic to $K_5$ and label its vertices as $W_i = \{ w_{ix}, w_{iy}, w_{iz}, w_{i}^+, w_{i}^- \}$.
    
    \item For each pair of $3$-element sets $c_i = (x,y,z)$ and $c_j = (a,b,c)$ such that $c_i, c_j \in  C$, add all edges between vertices of $\{ w_{ix}, w_{iy}, w_{iz} \}$ and $\{ w_{ja}, w_{jb}, w_{jc} \}$.
    
    \item For each element $x \in X$, generate a vertex $v_x$ and the edges $v_xw_{ix}$ for every $i$ such that $c_i$ contains the element $x$.
\end{enumerate}

Note that $G$ is indeed chordal, since a perfect elimination order can begin with the simplicial vertices $\{v_1, \dots, v_{3q}, w_1^+, \dots, w_m^+, w_1^-, \dots, w_m^-\}$, and be followed by an arbitrary sequence of the remaining vertices, which induce a clique.

An example of the reduction and its corresponding $c$-disconnected matching is presented in Figure \ref{fig:exemplo-c-disc}. For better visualization, the edges from rule (II) are omitted.

\begin{figure}
    \centering
    \includegraphics[width=1\textwidth]{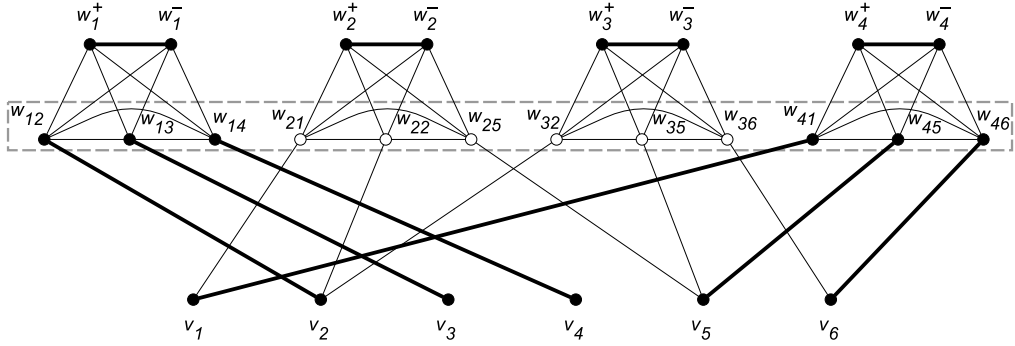}
    \caption{An example of reduction for the input $X = \{ 1,2,3,4,5,6 \}$ and $C = \{ \{ 2,3,4 \}, \{ 1,2,5 \}, \{ 2,5,6 \}, \{ 1,5,6 \} \}$. The subgraph induced by the vertices inside the dotted rectangle is complete and the matching in bold corresponds the solution $C' = \{ \{ 2,3,4 \}, \{ 1,5,6 \} \}$.}
    \label{fig:exemplo-c-disc}
\end{figure}

In Lemmas \ref{lemma:c-disc-ida} and \ref{lemma:c-disc-volta}, we define the polynomial transformation between a $(m-q+1)$-disconnected matching $M$, $|M| \geq m+3q$, and a subset $C'$ that solves the { \sc Exact Cover By 3-sets}. Then, Theorem \ref{teo:c-disc-npc-chordal} concludes the {\NPcness} for chordal graphs.

\begin{lemma}\label{lemma:c-disc-ida}
Let $(C,X)$ be an input of {\sc Exact Cover by 3-Sets} with $|C| = m$, $|X| = 3q$ and a solution $C'$. A $(m-q+1)$-disconnected matching $M$, $|M| = m+3q$, can be built in the transformation graph $G$ in polynomial time.
\end{lemma}
\begin{proof}
Denote the sets of vertices $U$ by $\{ w_i^+,w_i^- \mid c_i \in C \}$ and $S$ by $\bigcup_{j=1}^{j=m} V(H_j) - U$. Let's build a matching $M$ from the solution $C'$. For each set $c_i = \{ x,y,z \}$ contained in $C'$, add the edges $w_{ix}v_x$ to $M$. Also, for each set $c_j \in C$, add the edge $w_j^+w_j^-$. Consequently, each $H_i$ such that $c_i \notin C'$ will induce a connected component isomorphic to $K_2$ in $G[M]$, with the vertices $\{ w_i^+, w_i^- \}$, totalizing $m - q$ connected components and $2m-2q$ saturated vertices. There will also be one more connected component containing the $3q$ vertices of $\{ v_i \mid i \in [3q] \}$ and the $5q$ vertices of $\{ W_j \mid c_j \in C - C' \}$. Thus, $M$ saturates $2m - 2q + 3q + 5q = 2m + 6q$ vertices, corresponding to $m + 3q$ edges. Also, $G[M]$ has $m - q + 1$ connected components. So, $M$ is a valid solution for the { \sc $c$-Disconnected Matching }.
\end{proof}

\begin{lemma}\label{lemma:c-disc-volta}
Let $(C,X)$ be an input of {\sc Exact Cover by 3-Sets} with $|C| = m$, $|X| = 3q$. Given a $(m-q+1)$-disconnected matching $M$, $|M| = m+3q$, in the transformation graph $G$ described, a solution $C'$ to {\sc Exact Cover by 3-Sets} can be built in polynomial time.
\end{lemma}
\begin{proof}
Denote the sets of vertices $U$ by $\{ w_i^+,w_i^- \mid c_i \in C \}$ and $S$ by $\bigcup_{j=1}^{j=m} V(H_j) - U$. Consider an arbitrary $c$-disconnected matching $M$, $|M| \geq k$ in $G$. Based on the graph built, we show how such matching is structured and then build a solution $C'$. Note that every edge in $G$ is either incident to a vertex of $S$ or to two vertices of $U$. Since all vertices of $S$ are connected, $G[M]$ can only have two types of connected components. 

\begin{enumerate}[(I)]
    \item A $K_2$ with $2$ vertices of $U$.
    \item A connected component that can contain any vertex, except the ones from type (I) connected components and its adjacencies.
\end{enumerate}

Given that $G[M]$ has at least $c$ connected components, then it must have at least $c-1$ connected components of type (I). For each of them, the $2$ saturated vertices are $\{ w_i^+, w_i^- \}$, that are contained in $U$. To keep these vertices in an isolated connected component, the other $3$ vertices of $W_i \cap S$ can not be saturated. So, we will not use them for the rest of the construction. In order to have $c$ connected components, it is needed at least one more. Note that, so far, we have $2c-2 = 2m - 2q$ saturated vertices, and the remaining graph has exactly $6q + 2q = 8q$ vertices. Since $M$ saturates $2k = 2m-2q$ vertices, then the other connected component, of type (II), must saturate all the remaining $8q$ vertices. Thus, there can be no more than $c$ connected components in $G[M]$ and no more than $k$ edges in $M$. Note that, for each $i \in [m]$, the subgraph $G[W_i \cap S]$ has either $0$ or $3$ saturated vertices. If it has $3$, each vertex in $W_i \cap S$ must be matched with a vertex of $\{ v_j \mid j \in [3q] \}$, by an edge. Such edge exists because the set $c_i$ has the element $j$. Therefore, a solution $C'$ can contain the set $c_i$ if and only if $W_i \cap S$ has $3$ saturated vertices.
\end{proof}

\begin{theorem}\label{teo:c-disc-npc-chordal}
{\sc Disconnected Matching} is {\NPc} even for chordal graphs with diameter $2$.
\end{theorem}
\begin{proof}
Note that the $c$-disconnected matching is a certificate that the problem belongs to {\NP}. We now prove that it is also {\NPH}. The Lemmas \ref{lemma:c-disc-ida} and \ref{lemma:c-disc-volta} show that a solution $C'$ for the {\sc Exact Cover by $3$-sets} corresponds to a $c$-disconnected matching $M$ and vice versa, $c = m-q+1$, $|M| = 3q+m$ in the transformation graph $G$. Note that if we add an universal vertex to $G$, the same properties hold, and the diameter of $G$ is reduced to $2$. For this reason, { \sc $c$-Disconnected Matching } is {\NPH} and, thus, also {\NPc} even for chordal graphs.
\end{proof}

%We now prove, in sequence, that the problem belongs to {\NP} and is {\NPH}. Let $M$ be a $c$-disconnected matching, $|M| \geq k$. The matching $M$ serves as a certificate for the answer {\YES}, which can be verified in polynomial time, checking the size of $M$ and the number of connected components of $G[M]$. Then, the problem is in {\NP}.

We can also make little modifications to show that {\sc Disconnected Matching} is also {\NPc} for bounded vertex degree.

We can also show that the problem is hard even for limited vertex degree graphs. In this case, the rule (II) from the previous transformation can be replaced by the following.

\begin{enumerate}[(II)]
    \item For each pair of $3$-element sets $c_i = (x,y,z)$ and $c_j = (a,b,c)$ such that $c_i, c_j \in  C$, add all edges between vertices of $\{ w_{ix}, w_{iy}, w_{iz} \}$ and $\{ w_{ja}, w_{jb}, w_{jc} \}$ if and only if there is an element $p \in c_i,c_j$.
\end{enumerate}

We know that the problem \pname{Exact Cover By 3-Sets} remains {\NPc} even if no element in $X$ appears in more than $3$ sets of $C$\cite{garey_johnson}. Therefore, the transformation graph, though not chordal, will have the maximum degree bounded by a constant.

With the same arguments for the proof given in Theorem \ref{teo:c-disc-npc-chordal}, we enunciate the following theorem.

\begin{theorem}
{\sc Disconnected Matching} is {\NPc} even for graphs with bounded maximum degree and diameter $2$.
\end{theorem}
\section{Polynomial time algorithms}
For our final contributions, we turn our attention to positive results, showing that the problem is efficiently solvable in some graph classes.

\subsection{Minimal separators and disconnected matchings}\label{sec:minimal-separators}

It is not surprising that minimal separators play a role when looking for $c$-disconnected matchings.
In fact, for $c = 2$, Goddard et al.~\cite{GODDARD2005129} showed how to find 2-disconnected matchings in graphs with a polynomial number of minimal separators.
We generalize their result by showing that \pname{Disconnected Matching} parameterized by the number $c$ of connected components is in \XP; note that we do not need to assume that the family of minimal separators is part of the input, as it was shown in \cite{shen1997efficient} it can be constructed in polynomial time.

\begin{theorem}\label{teo:c-matching-xp}
\pname{Disconnected Matching} parameterized by the number of connected components is in {\XP} for graphs with a polynomial number of minimal separators.
\end{theorem}
\begin{proof}
Note that if a matching $M$ is a maximum $c$-disconnected matching of $G = (V,E)$, then there is a family $\mathcal{S}$ of at most $c-1$ minimal separators such that $V(G) - V(M)$ contains $\bigcup_{S \in \mathcal{S}} S$. Therefore, if we find such $\mathcal{S}$ that maximizes a maximum matching $M$ in $G[V - (\bigcup_{S \in \mathcal{S}} S)]$ and $M$ is $c$-disconnected, then $M$ is a maximum $c$-disconnected matching. Considering that $G$ has $|V|^{\bigO{1}}$ many minimal separators, the number of possible candidates for $\mathcal{S}$ is bounded by $|V|^{\bigO{c}}$. Computing a maximum matching can be done in polynomial time and checking whether $G[M]$ has $c$ components can be done in linear time. Therefore, the whole procedure takes $|V|^{\bigO{c}}$ time and finds a maximum $c$-disconnected matching.
\end{proof}

In particular, this result implies that \pname{$c$-Disconnected Matching} is solvable in polynomial time for chordal graphs~\cite{10.1007/3-540-44679-6_34}, circular-arc graphs~\cite{DEOGUN199939}, graphs that do not contain thetas, pyramids, prisms, or turtles as induced subgraphs~\cite{abrishami2020graphs}.
We leave as an open question to decide if \pname{Disconnected Matching} parameterized by $c$ is in \FPT\ for any of these classes.

\begin{comment}
A considerable effort has been put to study minimal separators, as they are related with some other problems, like potential maximal cliques and the maximum weight treewidth $k$ induced subgraph problem\cite{abrishami2020graphs}. In \cite{shen1997efficient}, it is shown that if a graph class has polynomial many separators, they can be enumerated in polynomial time. Also, \cite{abrishami2020graphs} shows that graph classes that do not contain a theta, pyramid, prism, or turtle as an induced subgraph have polynomial number of minimal separators.

Therefore, from Theorem \ref{teo:c-matching-xp}, we know that, for a fixed $c$, the { \sc $c$-Disconnected Matching } is in {\P} for chordal graphs\cite{10.1007/3-540-44679-6_34}, circular-arc graphs\cite{DEOGUN199939}, since they have polynomial number of minimal separators.

However, if we consider $c$ as part of the input, enumerating and testing all sets of $c$ minimal separators leads to algorithms that may not be polynomial, even for classes with polynomial number of minimal separators. This raises the question if { \sc $c$-Disconnected Matching } is hard for such classes.

In fact, adding $c$ to input turns the problem into {\NPc} for some classes, while it is still solvable in polynomial time for others. In the next sections, we show that { \sc $c$-Disconnected Matching } is {\NPc} even for chordal graphs and, next, that is in {\P} for interval and bounded treewidth graphs.
\end{comment}
\subsection{Interval Graphs}
\label{sec:disc-interval}

In this section, we show that {\sc Disconnected Matching} for interval graphs can be solved in polynomial time.
To obtain our dynamic programming algorithm, we rely on the ordering property of interval graphs~\cite{interval_order}; that is, there is an ordering $\mathcal{Q} = \langle Q_1, \dots Q_p \rangle$ of the $p$ maximal cliques of $G$ such that each vertex of $G$ occurs in consecutive elements of $\mathcal{Q}$ and, moreover the intersection $S_i = Q_i \cap Q_{i-1}$ between two consecutive cliques is a minimal separator of $G$.
Our algorithm builds a table $f(i,j,c')$, where $i,j \in [p]$ and $c' \in [c]$, and is equal to $q$ if and only if the largest $c'$-disconnected matching of $G\left[\bigcup_{i \leq \ell \leq j} Q_\ell \setminus (S_i \cup S_{j+1})\right]$ has $q$ edges; that is, $(G, k, c)$ is a positive instance if and only if $f(1, p, c) \geq k$.

\begin{theorem}\label{teo:c-disc-interval}
    \textsc{Disconnected Matching} can be solved in polynomial time on interval graphs.
\end{theorem}
\subsection{Treewidth}
\label{sec:tw}

A \textit{tree decomposition} of a graph $G$ is a pair $\mathbb{T} = \left(T, \mathcal{B} = \{B_j \mid j \in V(T)\}\right)$, where $T$ is a tree and $\mathcal{B} \subseteq 2^{V(G)}$ is a family where: $\bigcup_{B_j \in \mathcal{B}} B_j = V(G)$;
for every edge $uv \in E(G)$ there is some~$B_j$ such that $\{u,v\} \subseteq B_j$;
for every $i,j,q \in V(T)$, if $q$ is in the path between $i$ and $j$ in $T$, then $B_i \cap B_j \subseteq B_q$.
Each $B_j \in \mathcal{B}$ is called a \emph{bag} of the tree decomposition.
$G$ has treewidth at most $t$ if it admits a tree decomposition such that no bag has more than $t+1$ vertices.
For further properties of treewidth, we refer to~\cite{treewidth}.
After rooting $T$, $G_x$ denotes the subgraph of $G$ induced by the vertices contained in any bag that belongs to the subtree of $T$ rooted at node $x$.
Our final result is a standard dynamic programming algorithm on tree decompositions; we omit the proof and further discussions on how to construct the dynamic programming table for brevity.

An algorithmically useful property of tree decompositions is the existence of a \emph{nice tree decomposition} that does not increase the treewidth of $G$.

\begin{definition}[Nice tree decomposition]
    A tree decomposition $\mathbb{T}$ of $G$ is said to be \emph{nice} if its tree is rooted at, say, the empty bag $r(T)$ and each of its bags is from one of the following four types:
    \begin{enumerate}
        \item \emph{Leaf node}: a leaf $x$ of $\mathbb{T}$ with $B_x = \emptyset$.
        \item \emph{Introduce vertex node}: an inner bag $x$ of $\mathbb{T}$ with one child $y$ such that $B_x \setminus B_y = \{u\}$.
        \item \emph{Forget node}: an inner bag $x$ of $\mathbb{T}$ with one child $y$ such that $B_y \setminus B_x = \{u\}$.
        \item \emph{Join node}: an inner bag $x$ of $\mathbb{T}$ with two children $y,z$ such that $B_x = B_y = B_z$.
    \end{enumerate}
\end{definition}

\begin{comment}
An algorithmically useful property of tree decompositions is the existence of a \emph{nice tree decomposition} that does not increase the treewidth of $G$.

\begin{definition}[Nice tree decomposition]
    A tree decomposition $\mathbb{T}$ of $G$ is said to be \emph{nice} if its tree is rooted at, say, the empty bag $r(T)$ and each of its bags is from one of the following four types:
    \begin{enumerate}
        \item \emph{Leaf node}: a leaf $x$ of $\mathbb{T}$ with $B_x = \emptyset$.
        \item \emph{Introduce vertex node}: an inner bag $x$ of $\mathbb{T}$ with one child $y$ such that $B_x \setminus B_y = \{u\}$.
        \item \emph{Forget node}: an inner bag $x$ of $\mathbb{T}$ with one child $y$ such that $B_y \setminus B_x = \{u\}$.
        \item \emph{Join node}: an inner bag $x$ of $\mathbb{T}$ with two children $y,z$ such that $B_x = B_y = B_z$.
    \end{enumerate}
\end{definition}
\end{comment}

\begin{theorem}\label{teo:c-disc-tw}
    \textsc{Disconnected Matching} can be solved in \FPT\ time when parameterized by treewidth.
\end{theorem}
\begin{proof}
    We define $|V(G)| = n$ and $|E(G)| = m$.
    We suppose w.l.o.g. that we are given a tree decomposition $\mathbb{T} = (T, \mathcal{B})$ of $G$ of width $t$ rooted at an empty forget node; moreover, we solve the more general optimization problem, i.e., given $(G,c)$ we determine the size of the largest $c$-disconnected matching of $G$, if one exists, in time \FPT\ on $t$.
    As usual, we describe a dynamic programming algorithm that relies on $\mathbb{T}$.
    For each node $x \in V(T)$, we construct a table $f_x(S, U, \Gamma, \ell)$ which evaluates to $\rho$ if and only if there is a (partial) solution $M_x$ with the following properties: (i) $S \subseteq B_x \cap V(M_x)$ and $G[S]$ admits a perfect matching, (ii) the vertices of $U \subseteq V(M_x) \cap B_x \setminus S$ are half-matched vertices and are going to be matched to vertices in $G \setminus G_x$, (iv) $\Gamma \in \Pi(A_x)$ is a partition of $A_x$, where $S \cup U \subseteq A_x \subseteq B_x$, and each part of $\Gamma$ corresponds to a unique connected component of $G[M_x]$ --- we say that $V(\Gamma) = A_x$ --- (iv) $G[M_x]$ has \textit{exactly} $\ell$ connected components that do not intersect $B_x$, and (v) $M_x$ has $\rho$ edges.
    If no such solution exists, we define $f_x(S, U, \Gamma, \ell) = -\infty$ and we say the state is \textit{invalid}.
    Below, we show how to compute each entry for the table for each node type.
    
    \noindent \textbf{Leaf node:} Since $B_x = \emptyset$, the only valid entry is $f_x(\emptyset, \emptyset, \{\}, 0)$, which we define to be equal to 0.
    
    \noindent \textbf{Introduce node:} Let $y$ be the child of $x$ in $T$ and $B_x = B_y \cup \{v\}$.
    We compute the table as in Equation~\ref{eq:tw_introduce}; before proceeding, we define $\Gamma(v)$ to be the block of $\Gamma$ that contains $v$ and a partition $\Gamma' \in \sift_v(\Gamma)$ if: (i) $V(\Gamma') = V(\Gamma) \setminus \{v\}$, (ii) every block in $\Gamma \setminus \Gamma(v)$ is also in $\Gamma'$, and (iii) $\Gamma'$ contains a connected refinement of $\Gamma(v) \setminus \{v\}$, i.e. there is a partition $\pi$ of $\Gamma(v) \setminus \{v\}$ that is a subset of $\Gamma'$, vertices in different blocks of $\pi$ are non-adjacent, and each block of $\pi$ has at least one neighbor of $v$.
    
    \begin{equation}
        \label{eq:tw_introduce}
        f_x(S, U, \Gamma, \ell) =
        \begin{cases}
            f_y(S, U, \Gamma, \ell), &\text{ if $v \notin V(\Gamma)$;} \\
            \max\limits_{\Gamma' \in \sift_v(\Gamma)} f_y(S, U \setminus \{v\}, \Gamma', \ell), &\text{ if $v \in U$;} \\
            \max\limits_{u \in N(v) \cap S} \max\limits_{\Gamma' \in \sift_v(\Gamma)} f_y(S \setminus \{u,v\}, U \cup \{u\},  \Gamma', \ell), &\text{ otherwise.}
        \end{cases}
    \end{equation}
    
    For the first case of the above equation, if $v \notin V(\Gamma)$ then any partial solution $M_x$ of $G_x$ represented by $(S, U, \Gamma, \ell)$ is also a solution to $G_y$ under the same constraints since $G_y = G_x \setminus \{v\}$ and $v$ is not in $V(M_x) = V(\Gamma)$.
    For the second case, let $M_x$ by the solution that corresponds to $(S, U, \Gamma, \ell)$, $C_v$ the connected component of $G[V(M_x)]$ that contains $v$, and $ \pi = \{C_1, \dots, C_q\}$ the (possibly empty) connected components of $G[C_v \setminus \{v\}]$; by definition, $C_v \cap B_x = \Gamma(v)$ is a block of $\Gamma$ and $\pi$ is a partition of $C_v$ where vertices in different blocks are non-adjacent.
    Consequently, we have that $\Gamma' = \Gamma \setminus \Gamma(v) \cup \pi$ is in $\sift_v(\Gamma)$ and $f_y(S, U \setminus \{v\}, \Gamma', \ell)$ is accounted for in the computation of the maximum, which by induction is correctly computed.
    Finally, for the third case, let $uv \in M_x$, and note that we may proceed as in the previous case: the connected components of $G[C_v \setminus \{v\}]$ induce a connected refinement $\pi$ of $\Gamma(v) \setminus \{v\}$ and $\Gamma' = \Gamma \setminus \Gamma(v) \cup \pi$ is in $\sift_v(\Gamma)$, which results in a solution of $G_y$ that corresponds to the tuple $(S \setminus \{u,v\}, U \cup \{u\}, \Gamma', \ell)$; since the maximum runs over all neighbors of $v$ in $S$ and over all partitions in $\sift_v$, our table entry is correctly computed.

    \noindent \textbf{Forget node:} Let $y$ be the child of $x$ in $T$ and $B_x = B_y \setminus \{v\}$.
    We show how to compute tables for these nodes in Equation~\ref{eq:tw_forget}, where $\Gamma^{\{u,v\}} = \Gamma \setminus \{\Gamma(u)\} \cup \{\Gamma(u) \cup \{v\}\}$, $V_{S,U}(\Gamma) = V(\Gamma) \setminus (S \cup U)$, and
    
    \begin{equation}
        \label{eq:tw_forget_aux}
        g(S, U, \Gamma, \ell, u) = \max \{f_y(S, U, \Gamma^{\{u,v\}}, \ell),\ f_y(S \cup \{u,v\}, U, \Gamma^{\{u,v\}}, \ell) + 1\}.
    \end{equation}

    \begin{equation}
        \label{eq:tw_forget}
        f_x(S, U, \Gamma, \ell) = \max
        \begin{cases}
            f_y(S, U, \Gamma, \ell) \\
            f_y(S, U, \Gamma \cup \{\{v\}\}, \ell-1) \\
            \max\limits_{u \in V(\Gamma)} f_y(S, U, \Gamma^{\{u,v\}}, \ell), \text{ if $N(v) \cap V(\Gamma) = \emptyset$;} \\
            \max\limits_{u \in N(v) \cap \gamma} g(S, U, \Gamma, \ell, u), \text{ if $\exists \gamma \in \Gamma \mid N(v) \cap V(\Gamma) \subseteq \gamma$.}
        \end{cases}
    \end{equation}
    
    Let $M_x$ be a solution to $G_x$ represented by $(S, U, \Gamma, \ell)$.
    If $v \notin V(M_x)$, then $M_x$ is a solution to $G_y$ constrained by $(S, U, \Gamma, \ell)$ and, by induction, the correctness of $f_x$ is given by the first case of Equation~\ref{eq:tw_forget}.
    Recall that, assuming $v \in V(M_x)$ implies that $N(v) \cap V(\Gamma)$ must be contained in a single block of $\Gamma$, otherwise this table entry is deemed invalid and we may safely set it to $-\infty$.
    If there is some connected component $C_v$ of $G[V(M_x)]$ that has $C_v \cap (V(\Gamma) \cup \{v\}) = \{v\}$, then it must be the case that $M_x$ is a solution of $G_y$ represented by $(S, U, \Gamma \cup \{\{v\}\}, \ell - 1)$ since $C_v \cap B_y = \{v\}$, which is the second case of the equation.
    Suppose $wv \in M_x$.
    If $\{u,v\} \subseteq C_v \cap (V(\Gamma) \cup \{v\})$, we branch our analysis on two cases:
    \begin{enumerate}
        \item For the first one, we suppose $N(v) \cap V(\Gamma) = \emptyset$ and note that $\Gamma^{\{u,v\}}(u) = C_v \cap V(\Gamma^{\{u,v\}})$, so we must have that $f_y(S, U, \Gamma^{\{u,v\}}, \ell) \neq -\infty$ is accounted for in the third case of Equation~\ref{eq:tw_forget}.
        \item Otherwise, there is some $u \in N(v) \cap V(\Gamma)$ and it must be the case that $N(v) \cap V(\Gamma) \subseteq \Gamma(u) = \gamma$. If $u = w$, then $M_x$ is a partial solution to $G_y$ represented by $(S \cup \{u,v\}, U, \Gamma^{\{u,v\}}, \ell)$, so, by induction, $f_y(S \cup \{u,v\}, U, \Gamma^{\{u,v\}}, \ell)$ is well defined, and we have one additional matched edge outside of $B_x$ than outside of $B_y$, hence the $+1$ term in Equation~\ref{eq:tw_forget_aux}.
        Finally, if $u \neq w$ and $w \notin B_x$, then we proceed as in Case 1, as shown in Equation~\ref{eq:tw_forget_aux}, but since $v$ must be in the same connected component of its neighbors, we have fewer entries to check in $f_y$.
        Either way, $u$ is accounted for in the range of the maximum in the fourth case of Equation~\ref{eq:tw_forget}.
    \end{enumerate}
    
    \noindent \textbf{Join node:} Finally, let $x$ be a join node with children $y,z$ and $B_x = B_y = B_z$. We obtain the table for these nodes according to the following recurrence relation, where $C = A \sqcup B$ is the \textit{join} between $A$ and $B$.
    
    \begin{equation}
        \label{eq:tw_join}
        f_x(S, U, \Gamma, \ell) = \max_{\substack{\ell_y + \ell_z = \ell \\ U_y \dot{\cup} U_z = V_{S,U}(\Gamma) \\ \Gamma_y \sqcup \Gamma_z = \Gamma}} f_y(S, U \cup U_y, \Gamma_y, \ell_y) + f_z(S, U \cup U_z, \Gamma_y, \ell_z)
    \end{equation}
    
    Once again, let $M_x$ be a solution to $G_x$ satisfying $(S, U, \Gamma, \ell)$ and $M_i = M_x \cap G_i$ for $i \in \{y,z\}$.
    Moreover, let $\mathcal{C}_i$ be the connected components of $G[V(M_i)]$, $\ell_i$ the number of components in $\mathcal{C}_i$ with no vertex in $B_i$, and $\Gamma_i$ the partition of $V(M_i) \cap B_i$ where each block is equal to $C \cap B_i$ for some $C \in \mathcal{C}_i$.
    Note that it must be the case that $\ell = \ell_y + \ell_z$ --- since $M_x = M_y \cup M_z$ --- and that $\Gamma = \Gamma_y \sqcup \Gamma_z$ since vertices in different connected components of $M_y$ may be in a same connected component of $M_z$, but vertices in different connected components in both solutions are not merged in a same connected component of $M_x$.
    Now, define $W_y \subseteq V_{S,U}(\Gamma)$ to be the set of vertices that are matched to a vertex of $G_y \setminus S$, let $W_z$ be defined analogously, $U_y = W_z$ and $U_z = W_y$; note that $(U_y, U_z)$ is a partition of $V_{S,U}(\Gamma)$, and that the vertices in $M_i \cap B_i$ that must be matched, but not in $G_i$, are given by $U \cup U_i$.
    As such, $M_i$ is represented by $(S, U \cup U_i, \Gamma_i, \ell_i)$ and by induction we have $f_i(S, U \cup U_i, \Gamma_i, \ell_i) = |M_i \setminus S|$, so it holds that $f_x(S, U, \Gamma, \ell) = \sum_{i \in \{y,z\}} f_i(S, U \cup U_i, \Gamma_i, \ell_i) = |M_x \setminus S|$, which is one of the terms of the maximum shown in Equation~\ref{eq:tw_join}.
    
    Recall that we may assume that our tree decomposition is rooted at a forget node $r$ with $B_r = \emptyset$; by definition, our instance $(G, k, c)$ of \textsc{Disconnected Matching} is a \YES\ instance if and only if $f_r(\emptyset, \emptyset, \{\}, c) \geq k$.
    As to the running time, we have $\bigO{2^{2t}\eta_tn}$ entries per table of our algorithm, where $\eta_t$ be the $t$-th Bell number, each of which can be computed in $\bigO{2^t\eta_{t+1}^2n}$, which is the complexity of computing a join node, so our final running time is of the other of $\bigO{8^t\eta_{t+1}^3n^2}$.
\end{proof}
\section{Kernelization}

In the previous section, we presented an \FPT\ algorithm for the treewidth parameterization, which implies tractability for several other parameters, such as vertex cover and max leaf number.
In this section, we provide kernelization lower bounds for \pname{Disconnected Matching} when parameterized by vertex cover and when parameterized by vertex deletion distance to clique.
We highlight that our lower bounds hold for the \pname{Induced Matching} problem and, for the former parameterization, even when restricted to bipartite graphs.
Our proofs are through OR-cross-compositions from the \pname{Exact Cover by 3-Sets} problem, and are inspired by the proof of Section~\ref{sec:c-disc-chordal}.
Throughout this section, let $\mathcal{S} = \{(X_1, C_1), \dots, (X_t, C_t)\}$ be the input instances to \pname{Exact Cover by 3-Sets}; w.l.o.g., we assume that $X_i = X = [n]$ and $|C_i| = m$ for all $i \in [t]$, and define $\mathcal{C} = \bigcup_{i \in [t]} C_i$.
We further assume that, for any two instances, $C_i \neq C_j$, which implies that $C_i \setminus C_j$ and $C_j \setminus C_i$ are non-empty.
We denote by $(G,k)$ the built \pname{Induced Matching} instance.

\subsection{Vertex Cover}

\noindent\textbf{Construction.}
We begin by adding to $G$ the set $A = \{v_a \mid a \in X\}$ and, for each set $S_j \in \mathcal{C}$ where $S_j = \{a,b,c\}$, we add one copy $Q_j$ of $K_{1,4}$, with vertices $\{w_j, w_j^*, w_{ja}, w_{jb}, w_{jc}\}$; $w_j$ is the central vertex, while $w_{ja}, w_{jb}, w_{jc}$ are the \textit{interface} vertices of $Q_j$. Then, we add edges to $G$ so $w_{ja}v_d \in E(G)$ if and only if $a = d$.
Now, we add to $G$ an instance selector gadget $I$, which is simple a star with $t$ leaves, with the central vertex labeled as $q$ and the $i$-th leaf labeled as $p_i$.
To complete the construction of $G$, for each $p_i \in V(I)$ and $S_j \in \mathcal{C} \setminus C_i$, we add all edges between $p_i$ and the interface vertices of $Q_j$, i.e. if $S_j$ is not a set of the $i$-th instance, we add edges between $Q_j$ and $p_i$.
Finally, we set $k = n + |\mathcal{C}| - \frac{n}{3} + 1$.

\begin{lemma}
    \label{lem:param_bound_nokernel_vc}
    Graph $G$ is bipartite and has a vertex cover of size $\bigO{n^3}$.
\end{lemma}

\begin{proof}
    We construct the bipartition $(Y,W)$ as follows: $Y = (I \setminus \{q\}) \cup \{w_j \mid S_j \in \mathcal{C}\} \cup A$ and $W = V(G) \setminus Y$.
    To see that $Y$ is an independent set, note that: (i) each of its three components induce independent sets in $G$, (ii) $I \setminus \{q\}$ is not adjacent to the central vertex of any $Q_j$ nor to any vertex of $A$, and (iii) vertices of $A$ are non-adjacent to the central vertices of the $Q_j$'s.
    For $W$, note it is composed by the leaves of the $Q_j$'s, which together form an independent set, and vertex $q$, which is only adjacent to vertices of $I$, none of which belong to $W$.
    Note $D = V(G) \setminus (V(I) \setminus \{q\})$ is a vertex cover, i.e. $G \setminus D$ is an independent set, since each connected component of $G \setminus D$ is corresponds to a leaf of $I$.
    Observe that $|D| = 5|\mathcal{C}| + n + 1$ and that there are most $\binom{n}{3}$ elements in $\mathcal{C}$ since there are at most this many subsets of three distinct elements of the ground set $X$, so $|D| \in \bigO{n^3}$.
\end{proof}

\begin{lemma}
    \label{lem:forward_nokernel_vc}
    If $(X_\ell, C_\ell)$ admits a solution, then $G$ admits an induced matching with $k$ edges.
\end{lemma}

\begin{proof}
    Let $\Pi$ be the solution to $(X_\ell, C_\ell)$, $S_j \in \mathcal{C}$, and $S_j = \{a,b,c\}$.
    We add to $M$ the edges $\{w_{ja}v_a, w_{jb}v_b, w_{jc}v_c\}$, if $S_j \in \Pi$, otherwise we add edge $w_jw_j^*$ to $M$, totalling $n + |\mathcal{C}| - \frac{n}{3}$ edges.
    For the final edge, add $qp_\ell$ to $M$.
    In terms of connected components, each edge of $M$ is a distinct component, since: (i) each $Q_j$ either has 3 of its leaves in $M$ but not its central vertex, or it has its central vertex in $M$, (ii) each vertex of $A$ is adjacent to only one saturated vertex, i.e. its only neighbor in $V(M)$ is $w_{ja}$, and (iii) $p_\ell$ is adjacent only to interface vertices that are not saturated by $M$, so its unique neighbor in $V(M)$ is $q$.
    As such, $M$ is an induced matching with $k$ edges.
\end{proof}

Let us now show the converse.

\begin{lemma}
    \label{lem:normalized_nokernel_vc}
    In every solution $M$ of $(G,k)$, $q$ is $M$-saturated.
\end{lemma}

\begin{proof}
    Towards a contradiction, suppose that $q \notin V(M)$ and, furthermore that $I \cap V(M) = \emptyset$.
    In this case, note that $|M| \leq n + |\mathcal{C}| - \frac{n}{3} = k - 1$, since we may have at most $\frac{n}{3}$ stars $Q_j$ with the three interface vertices in $M$, contributing with $n$ edges to $M$, and all other $Q_j$'s have at most edge $w_jw_j^*$ in $M$, totaling $n + |\mathcal{C}| - \frac{n}{3}$ edges in the matching.
    
    If, on the other hand, $I \cap V(M) \neq \emptyset$, then suppose $p_i \in V(M)$.
    Since $q \notin V(M)$, $p_i$ is matched with a vertex in $Q_j$, say $w_{ja}$, which implies that $Q_j \cap V(M) = \{w_{ja}\}$, since $p_i$ is adjacent to all three interface vertices of $Q_j$ and, if $w_j^*$ is saturated by $M$, then $w_j$ also is, which is impossible since $M$ is an induced matching.
    Moreover, note that, for every $Q_x$ with vertices adjacent to $p_i$, we have that $E(Q_x) \cap M \subseteq \{w_xw_x^*\}$.
    At this point, we have accounted for $1 + (|\mathcal{C}| - m - 1) =  |\mathcal{C}| - m$ edges of $M$.
    For the $m$ $Q_j$'s with no vertex adjacent to $p_i$, they can each contribute with at most three edges to $M$ but no more than $n + m - \frac{n}{3}$ in total, since each $Q_j$ will either: (i) have some of its interface vertices matched to vertices $\{v_a, v_b, v_c\}$, (ii) have $w_j^*$ saturated, or (iii) have exactly one of its interface vertices saturated to either some other $p_y$ or to $w_j$.
    As such, we have at most $n$ edges coming from the first option, while the others amount to, at most $m - \frac{n}{3}$ additional edges.
    Finally, this implies that $|M| \leq |\mathcal{C}| - m + n + m - \frac{n}{3} = n + |\mathcal{C}| - \frac{n}{3} < k$, and we conclude that $q$ must be $M$-saturated.
\end{proof}

As a consequence of our previous lemma, there is an edge of the form $qp_i$ in every solution to $(G,k)$.

\begin{lemma}
    \label{lem:backward_nokernel_vc}
    If $(G,k)$ admits a solution, then at least one instance $(X_\ell, C_\ell) \in \mathcal{S}$ also admits a solution.
\end{lemma}

\begin{proof}
    Let $M$ be a solution to $(G,k)$ with $qp_\ell \in M$ and $\mathcal{Q}$ be the set of $Q_j$'s with at least one saturated interface vertex.
    Note that no vertex in $Q_j \in \mathcal{Q}$ is adjacent to $p_\ell$, otherwise $V(M)$ would not induce a matching.
    Let us show that $|\mathcal{Q}| = \frac{n}{3}$.
    If we had any more elements in $\mathcal{Q}$, $M$ would have at most $n$ edges incident to an interface vertex and at most $|\mathcal{C}| - |\mathcal{Q}|$ edges incident to the non-interface vertices of $Q_j$'s, which implies that $|M| \leq 1 + n + |\mathcal{C}| - |\mathcal{Q}| < 1 + n + |\mathcal{C}| - \frac{n}{3} = k$; the first property follows from the fact that $q$ is already in $V(M)$ and the neighbors of interface vertices, aside outside of $A$ and the central vertex of $Q_j$, are also neighbors of $q$.
    On the other hand, if $|\mathcal{Q}| < \frac{n}{3}$, then we would have that $|M| \leq 1 + 3|\mathcal{Q}| + |\mathcal{C}| - |\mathcal{Q}| < 1 + |\mathcal{C}| + 2\frac{n}{3} = k$.
    With this in hand, note that, to obtain $|M| = k$, it must be the case that $|V(M) \cap \bigcup_{Q_j \in \mathcal{Q}} Q_j| = n$ and interface vertices are matched with vertices of $A$. This holds since $k \geq |M| = |M \cap E(I)| + \sum_{S_j \in \mathcal{C}} |M \cap E(Q_j)| \leq 1 + 3|\mathcal{Q}| + |\mathcal{C}| - |\mathcal{Q}| = 1 + n + |\mathcal{C}| - \frac{n}{3} = k$.
    Moreover, the elements of $\mathcal{Q}$ must not be adjacent to $p_\ell$, which implies that, for each $Q_j \in \mathcal{Q}$, we have that $S_j \in C_\ell$.
    Since vertices $\{a,b,c\}$ of $A$ are matched with vertices $\{w_{ja}, w_{jb}, w_{jc}\}$ of $Q_j \in \mathcal{Q}$, it follows that $S_j = \{a,b,c\}$ and that $\{S_j \mid Q_j \in \mathcal{Q}\}$ is a solution to $(X_\ell, C_\ell)$.
\end{proof}

Finally, combining Lemmas~\ref{lem:param_bound_nokernel_vc},~\ref{lem:forward_nokernel_vc}, and~\ref{lem:backward_nokernel_vc}, we obtain our kernelization lower bound.

\begin{theorem}
    \label{thm:no_kernel_vc}
    \pname{Induced Matching} does not admit a polynomial kernel when jointly parameterized by vertex cover and solution size unless $\NP \subseteq \coNP/\poly$, even when restricted to bipartite graphs.
\end{theorem}

\begin{corollary}
    \pname{Disconnected matching} does not admit a polynomial kernel when jointly parameterized by vertex cover and number of edges in the matching unless $\NP \subseteq \coNP/\poly$, even when restricted to bipartite graphs.
\end{corollary}

\subsection{Distance to Clique}

It is worthy to note at this point that the proof we have just presented can be adapted to the distance to clique parameterization without significant changes.
To do so, we replace $I$ with a clique of size $t + 1$, label its vertices arbitrarily as $\{q, p_1, \dots, p_t\}$ and proceed exactly as before.
The caveat being that we must show that any solution that picks an edge $p_xp_y$ can be changed into a solution that picks, say, $qp_x$ and that this new solution behaves in the exact same way as the one we outline in Lemma~\ref{lem:backward_nokernel_vc}.
We prove this in the following lemma.

\begin{lemma}
    \label{lem:normalized_nokernel_dc1}
    If $(G,k)$ admits a solution, then it also admits a solution where $q$ is saturated.
\end{lemma}

\begin{proof}
    Let $M$ be a solution to $(G,k)$ that does not saturate $q$.
    Our first claim is that $M \cap E(I) \neq \emptyset$.
    Note that no $p_i$ may be matched to a vertex outside of $I$; we could immediately repeat the second paragraph of the proof of Lemma~\ref{lem:normalized_nokernel_vc}, so if $M \cap E(I) \neq \emptyset$, it holds that $V(M) \cap I = \emptyset$.
    Now, observe that $|M| \leq k - 1$ since the maximum induced matching in $G \setminus I$ uses as many edges between $A$ and the $Q_j$'s as possible and, for each $Q_j$ without $M$-saturated interface vertices, we pick edge $w_jw_j^*$, totalling at most $|A| + |\mathcal{C}| - \frac{n}{3} \leq k - 1$ edges in $M$.
    As such, $p_xp_y \in M$ these are the only vertices saturated in $I$, otherwise $M$ would not be induced.
    Replacing edge $p_xp_y$ by edge $qp_y$ maintains the property that $M$ is an induced matching and does not change its cardinality, completing the proof.
\end{proof}

At this point, we can immediately repeat the proof of Lemma~\ref{lem:backward_nokernel_vc}. Together with  Lemma~\ref{lem:normalized_nokernel_dc1}, we observe that $(G,k)$ admits a solution if and only if some instance $(X, C_i)$ also does.
Finally, by observing that the same vertex cover described in Lemma~\ref{lem:param_bound_nokernel_vc} is a clique modulator for the current construction, we obtain the following theorem.

\begin{theorem}
    \label{thm:no_kernel_dc1}
    \pname{Induced Matching} does not admit a polynomial kernel when jointly parameterized by vertex deletion distance to clique and solution size unless $\NP \subseteq \coNP/\poly$.
\end{theorem}

\begin{corollary}
    \pname{Disconnected matching} does not admit a polynomial kernel when jointly parameterized by vertex deletion distance to clique and number of edges in the matching unless $\NP \subseteq \coNP/\poly$.
\end{corollary}

\section{Conclusions and future works}\label{sec:conclusions}

We have presented $c$-disconnected matchings and the corresponding decision problem, which we named \pname{Disconnected Matching}.
They generalize the well studied induced matchings and the problem of recognizing graphs that admit a sufficiently large induced matching.
Our results show that, when the number of connected components $c$ is fixed, \pname{$c$-Disconnected Matching} is solvable in polynomial time if $c = 1$ but \NPc\ even on bipartite graphs if $c \geq 2$.
We also proved that, unlike \pname{Induced Matching}, \pname{Disconnected Matching} remains \NPc\ on chordal graphs.
On the positive side, we show that the problem can be solved in polynomial time for interval graphs, in \XP\ time for graphs with a polynomial number of minimal separators when parameterized by the number of connected components $c$, and in \FPT\ time when parameterized by treewidth.
Finally, we showed that \pname{Disconnected Matching} does not admit polynomial kernels for very powerful parameters, namely vertex cover and vertex deletion distance to clique, by showing that this holds for the \pname{Induced Matching} particular case.

Possible directions for future work include determining the complexity of the problem on different graph classes.
In particular, we would like to know the complexity of \pname{Disconnected Matching} for strongly chordal graphs; we note that the reduction presented in Section~\ref{sec:c-disc-chordal} has many induced subgraphs isomorphic to a sun graph.

Aside from graph classes, we would like to understand structural properties of disconnected matchings. In particular, we are interested in determining sufficient conditions for a graph $G$ to have $\beta_{d,2}(G) = \beta_{*}(G)$ or $\beta_{d,2}(G) = \beta(G)$.

We are also interested in the parameterized complexity of the problem.
Our results show that, when parameterized by $c$, the problem is {\pNPH}; on the other hand, it is {\WH{1}} parameterized by the number of edges in the matching since \pname{Induced Matching} is {\WH{1}} under this parameterization~\cite{MOSER2009715}.
A first question of interest is whether chordal graphs admit an \FPT\ algorithm when parameterized by $c$; while the algorithm presented in Section~\ref{sec:minimal-separators} works for all classes with a polynomial number of minimal separators, chordal graphs offer additional properties that may aid in the proof of an \FPT\ algorithm.
Another research direction would be the investigation of other structural parameterizations, such as vertex cover and cliquewidth; while the former yields a fixed-parameter tractable algorithm due to Theorem~\ref{teo:c-disc-tw}, we would like to know if we can find a single exponential time algorithm under this weaker parameterization.
On the other hand, cliquewidth is a natural next step, as graphs of bounded treewidth have bounded cliquewidth, but the converse does not hold.
Finally, while we have settled several kernelization questions for \pname{Disconnected Matching} and \pname{Induced Matching}, other parameterizations are still of interest, such as max leaf number, feedback edge set, and neighborhood diversity.

We are currently working on weighted versions of $\mathscr{P}$-matchings, where we want to find matchings whose sum of the edge weights is sufficiently large, and the subgraph induced by the vertices of the matching satisfies some given property.

\section*{Acknowledgements} We thank the research agencies CAPES, CNPq, FAPEMIG, and FAPERJ for partially funding this work.

%% If you have bibdatabase file and want bibtex to generate the
%% bibitems, please use
%%
 \bibliographystyle{elsarticle-num} 
 \bibliography{cas-refs}

%% else use the following coding to input the bibitems directly in the
%% TeX file.

% \begin{thebibliography}{00}

% %% \bibitem{label}
% %% Text of bibliographic item

% \bibitem{}

% \end{thebibliography}
\end{document}